\documentclass{article}
\newif\ifhidecomments
\hidecommentstrue

\usepackage{ifthen}
\usepackage{xspace}
\usepackage{amsmath, amssymb, amsfonts, amsthm}
\usepackage[pdftex]{graphicx}
\usepackage[english]{babel}
\usepackage{color}
\usepackage{verbatim}
\usepackage{subfig}
\usepackage{float}
\usepackage{algorithmic}
\usepackage{hyperref}
\usepackage{algorithm}

\allowdisplaybreaks

\newtheorem{theorem}{Theorem}

\newtheorem{corollary}{Corollary}
\newtheorem{proposition}{Proposition}
\newtheorem{lemma}{Lemma}
\newtheorem{remark}{Remark}

\theoremstyle{definition}
\newtheorem{definition}{Definition}

% ...

% ...

\providecommand{\Expect}[2][]{\ensuremath{%
\ifthenelse{\equal{#1}{}}{\mathbf{E}}{\mathbf{E}_{#1}}%
\left[#2\right]}\xspace}

\renewcommand{\varnothing}{\ensuremath{\emptyset}}

\ifhidecomments
\newcommand{\tkedit}[1]{}
\newcommand{\tk}[1]{{#1}}
\newcommand{\bkedit}[1]{}
\newcommand{\bk}[1]{{#1}}
\newcommand{\yzedit}[1]{}

\newcommand{\remove}[1]{}
\else
\newcommand{\tkedit}[1]{{\color{red}$[\bullet]$}{\marginpar{\scriptsize\color{red}[GS: #1]}}}
\newcommand{\tk}[1]{{\color{red} #1}}
\newcommand{\bkedit}[1]{{\color{magenta}$[\bullet]$}{\marginpar{\scriptsize\color{magenta}[BK: #1]}}}
\newcommand{\bk}[1]{{\color{magenta} #1}}
\newcommand{\yzedit}[1]{{\color{green}$[\bullet]$}{\marginpar{\scriptsize\color{blue}[PC: #1]}}}

\newcommand{\remove}[1]{{\color{yellow} #1}}
\fi

\newcommand{\Omit}[1]{}

\newenvironment{proof-sketch}{\noindent{\textit{Proof sketch.}}}{\smallskip}

\title{Solving Weighted Voting Game \\ Design Problems Optimally: \\ Representations, Synthesis, and Enumeration}

\author{Bart de Keijzer\thanks{Algorithms, Combinatorics and Optimization; Centrum Wiskunde \& Informatica; The Netherlands; Email: \texttt{keijzer@cwi.nl}.} \and Tomas B. Klos\thanks{Algorithmics; Delft University of Technology; The Netherlands; Email: \texttt{T.B.Klos@tudelft.nl}.} \and Yingqian Zhang\thanks{Department of Econometrics; Erasmus University Rotterdam; The Netherlands; Email: \texttt{yqzhang@ese.eur.nl}}}
\date{}
\begin{document}
\pagestyle{plain}
\bibliographystyle{plain}
\sloppy
\maketitle
\begin{abstract}
We study the \emph{power index voting game design problem} for weighted voting games: the problem of finding a weighted voting game in which the power of the players is as close as possible to a certain target distribution. Our goal is to find algorithms that solve this problem exactly. Thereto, we consider various subclasses of simple games, and their associated representation methods. We survey algorithms and impossibility results for the \emph{synthesis problem}, i.e., converting a representation of a simple game into another representation.

We contribute to the synthesis problem by showing that it is impossible to compute in polynomial time the list of \emph{ceiling coalitions} (also known as \emph{shift-maximal losing coalitions}) of a game from its list of \emph{roof coalitions} (also known as \emph{shift-minimal winning coalitions}), and vice versa.

Then, we proceed by studying the problem of enumerating the set of weighted voting games. We present first a naive algorithm for this, running in doubly exponential time. Using our knowledge of the synthesis problem, we then improve on this naive algorithm, and we obtain an enumeration algorithm that runs in quadratic exponential time (that is, $O(2^{n^2} \cdot p(n))$ for a polynomial $p$). Moreover, we show that this algorithm runs in \emph{output-polynomial time}, making it the best possible enumeration algorithm up to a polynomial factor.

Finally, we propose an exact \emph{anytime} algorithm for the power index voting game design problem that runs in exponential time. This algorithm is straightforward and general: it computes the error for each game enumerated, and outputs the game that minimizes this error. By the genericity of our approach, our algorithm can be used to find a weighted voting game that optimizes any exponential time computable function. We implement our algorithm for the case of the normalized Banzhaf index, and we perform experiments in order to study performance and error convergence.
\end{abstract}

\section{Introduction}\label{intro}

In many real-world problems that involve multiple agents, for instance elections, there is a need for fair decision making protocols in which different agents have different amounts of influence in the outcome of a decision. \textit{Weighted voting games} are often used in these decision making protocols. In a weighted voting game, a quota is given, and each agent (or also: player) in the game has a certain weight. If the total weight of a coalition of agents exceeds the quota, then that coalition is said to be \emph{winning}, and \emph{losing} otherwise.

Weighted voting games arise in various settings, such as political decision making (decision making among larger and smaller political parties), stockholder companies (where people with different numbers of shares are supposed to have a different amount of influence), and elections (e.g., in the US Presidential Election, where each state can be regarded as a player who has a weight equal to its number of electors).

The weight that a player has in a weighted voting game turns out not to be equal to his actual influence on the outcome of the decisions that are made using the weighted voting game. Consider for example a weighted voting game in which the quota is equal to the sum of the weights of all players. In such a game, a player's influence is equal to the influence of any other player, no matter what weight he has. Throughout the literature, various \emph{power indices} have been proposed: ways to measure a player's influence (or \textit{(a priori) power}) in a voting game. However, computing a power index turns out to be a challenge in many cases.

In this paper, instead of analyzing the power of each agent in a voting game, we investigate the problem that has been referred to as the ``inverse problem'' and the ``generalized apportionment problem''. We will call this problem the \emph{power index voting game design problem}. In the power index voting game design problem we are given a target power index for each of the agents, and we study how to design a weighted voting game for which the power of each agent is as close as possible to the given target power index. The power index voting game design problem is an instantiation of a larger class of problems that we will call \emph{voting game design problems}, in which the goal is to find a game $G$ in a class of voting games such that $G$ has a given set of target properties. 
cite{azizinverse}, who propose an algorithm to find a WVG for the Banzhaf index.

The practical motivation behind our work is obvious: It is desirable to have an algorithm that can quickly compute a fair voting protocol, given that we want each agent to have some specified amount of influence in the outcome. When new decision making bodies must be formed, or when changes occur in the formation of these bodies, such an algorithm may be used to design a voting method that is as fair as possible. Only little work is known that tries to solve this problem. Two of the proposed algorithms (see~\cite{shaheeninverse,azizinverse}) are local search methods that do not guarantee an optimal answer. There is one other paper by Kurz \cite{kurzinverse} that proposes a method for finding an exact algorithm. Such an algorithm to solve the inverse problem exactly is also the topic of this paper: We are interested in finding a game for which the power index of that game is the \textit{closest possible} to a certain target power index. The work we present here is independent from \cite{kurzinverse}, and 
differs from it in the sense that we put more emphasis on run-time analysis and giving proofs of various desirable properties of our algorithm. Moreover, our approach is vastly different from the approach of \cite{kurzinverse}, and the theory behind our algorithm is of independent interest.

It seems that the most straightforward approach to solve the inverse problem would be to simply enumerate \emph{all possible} weighted voting games of $n$ players, and to compute for each of these weighted voting games its power index. We can then output the game of which the power index is the closest to the given target power index. This is precisely what we will do in this paper. Unfortunately, it turns out that enumerating all weighted voting games efficiently is not so straightforward.

The enumeration method that we present in this paper leads to a \emph{generic exponential time exact anytime algorithm} for solving voting game design problems. We implemented our algorithm for the power index voting game design problem where our power index of choice is the \emph{(normalized) Banzhaf index}: one of the two most widely used power indices. Using this implementation, we experimentally study the runtime and error convergence of this algorithm.

\subsection{Contributions}

This paper is based on the master's thesis of De Keijzer \cite{dekeijzer}; one of the authors of this manuscript. A shorter discussion of this work has also appeared \cite{dekeijzer2}. We present and discuss the results of \cite{dekeijzer}, and remove various redundancies, imprecisions, typos, and mistakes that were present in \cite{dekeijzer}. The results we present are thus as follows:
\begin{itemize}
\item We provide a general definition of voting game design problems, and show how the power index voting game design problem is one of these problems.
\item We present lower and upper bounds on the cardinalities of various classes of games. As it turns out, for many of these classes there is a very strong connection with certain classes of boolean functions; allowing us to borrow many bounds directly from boolean function theory.
\item We investigate thoroughly the problem of transforming various representations for simple games into each other. We give an overview of known results, and we present a new result: We prove that it is not possible to transform within polynomial time the \emph{roof-representation} of a game into a \emph{ceiling-representation}, and vice versa.\footnote{The roof-representation is also known as the \emph{shift-minimal winning coalition} representation \cite{simplegames}. Likewise, the ceiling representation is also known as the \emph{shift-minimal losing coalition} representation.}
\item We present exact algorithms for solving power index voting game design problems: first, a doubly exponential one for the large class of monotonic simple games; and subsequently we show that it is possible to obtain a (singly) exponential algorithm for the important special case of weighted voting games. This can be regarded as the main result of this paper.
\begin{itemize}
\item At the core of these algorithms lie methods for enumerating classes of games. Therefore, it actually follows that the same approach can be used for solving practically any voting game design problem.
\item The method that we use for enumerating weighted voting games is based on a new partial order on the class of weighted voting games, that has some specific interesting properties. This result is of independent interest from a mathematical point of view.
\end{itemize}
\item The algorithm for solving the power index voting game design problem for the case of weighted voting games (mentioned in the previous point) is based on working with families of minimal winning coalitions. We show how it is possible to improve the runtime of this algorithm by showing that it suffices to only work with a subset of these minimal winning coalitions: The \emph{roof coalitions}. Using this idea, we provide various techniques to improve our algorithm. Among these improvements is an output-polynomial time algorithm for outputting the list of ceiling coalitions of a \emph{linear game}, given the list of roof coalitions.
\item Finally, we implement the aforementioned enumeration algorithm for weighted voting games, in order to measure its performance, obtain some interesting data about the class of weighted voting games, and validate some theoretical results related to weighted voting games.
\end{itemize}

\subsection{Related work}\label{knownmethods}

Although some specific variants of voting game design problems are mentioned sporadically in the literature, not many serious attempts to solve these problems are known to us. The voting game design problem that is usually studied is that of finding a weighted voting game, represented as a weight vector, for which the power index lies as close as possible to a given target power index. This specific version of the voting game design problem is sometimes referred to as \textit{the inverse problem}, and is also the focus of this paper.

We know of only a few papers where the authors propose algorithms for the inverse problem. One of them is by Fatima et al.\ \cite{shaheeninverse}, where the authors present an algorithm for the inverse problem with the Shapley-Shubik index \cite{shapleyshubik} as the power index of choice. This algorithm works essentially as follows: It first receives as input a target Shapley-Shubik index and a vector of initial weights. After that, the algorithm enters an infinite loop where repeatedly the Shapley-Shubik index is computed, and the weight vector is updated according to some rule. The Shapley-Shubik index is computed using a linear time randomized approximation algorithm, proposed in \cite{shaheenapprox1,shaheenapprox2} by the same authors. For updating the weights, the authors propose two different rules of which they prove that by applying them, the Shapley-Shubik index of each player cannot get worse. Hence, the proposed algorithm is an \textit{anytime} algorithm: it can be terminated at any time, but
gets closer to the optimal answer the longer the algorithm runs. No analysis on the approximation error is done, although the authors mention in a footnote that analysis will be done in future work. The runtime of one iteration of the algorithm is shown to be $O(n^2)$ (where $n$ denotes the number of players).

Another algorithm is by Aziz et al.\ \cite{azizinverse} for the inverse problem with the Banzhaf index as the power index of choice.
The algorithm the authors present here resembles that of \cite{shaheeninverse}, in that the algorithm repeatedly updates the weight vector in order to get closer to the target power index. The algorithm gets as input a target Banzhaf index. As an initial step, an integer weight vector is estimated according to a normal distribution approximation. Subsequently, the algorithm enters an infinite loop, and consecutively computes the Banzhaf index and updates the weight. For computing the Banzhaf index, the \textit{generating function method} is used \cite{bilbaofernandez,generatingshapley,generatingbrams}. This is an exact pseudopolynomial time method that works only when the weights in the weighted representation of a game are integers. Therefore, the output of the algorithm is always an integer weighted representation (contrary to the method in  \cite{shaheeninverse} for which the output may have rational weights). The updating is done by interpolating a best fit curve. This results in a rational weight vector.
 To obtain integer weights, the weight vector is rounded to integers, but prior to that it is multiplied by a suitable constant that reduces the error when rounding to integers.

For Aziz's approach, there is no approximation guarantee and the convergence rate is unknown, so it is not certain whether this method is \textit{anytime} just like Fatima's algorithm. Moreover, not much is known about the time complexity and practical performance of this algorithm (one example is presented of this algorithm working on a specific input).

Leech proposes in \cite{leechsurvey,leechdesign} an approach that largely resembles the method of Aziz et al.: it is the same, with the exception that a different updating rule is used. The method that Leech uses for computing the Banzhaf index is not mentioned. The focus in this paper is on the results that are obtained after applying the method to the 15-member EU council (also see \cite{leechdesigneu}), and to the board of governors of the International Monetary Fund.

There are two more recent interesting works on the voting game design problem. One is by Kurz \cite{kurzinverse}. Kurz proposes an exact method using integer linear programming, for solving the weighted voting game design problem for both the Shapley-Shubik index and the Banzhaf index.
The set of linear games is taken as the search space, and branch-and-bound techniques (along with various insights about the set of weighted voting games) are used in order to find in this set a weighted voting game with a power index closest to the target.
Kurz does not provide a runtime analysis. The experiments performed show that the algorithm works well for small numbers of players.
As mentioned in Section \ref{intro}, our work is independent and differs from \cite{kurzinverse} because we are interested in devising an algorithm with a provably good runtime. Moreover, the approach we take is different from that of \cite{kurzinverse}, and
the theory necessary to develop our algorithm can be considered interesting in itself.

Kurz moreover correctly points out that in the master's thesis of De Keijzer \cite{dekeijzer} (on which the present paper is based) the numbers of canonical weighted voting games for $6$, $7$, and $8$ players are wrongly stated. After investigation on our part, it turned out that this is due to a bug in the first implementation of the algorithm. In this paper, we correct this mistake and report the numbers of canonical weighted voting games correctly, although these numbers are already known by now due to the recent paper \cite{kurzminsum}, also by Kurz (see below).

The other recent work is \cite{servedioapprox}, by De et al. This paper gives provides as a main result an algorithm for the inverse power index problem for the case of the Shapley-Shubik index, and has a certain approximation guarantee: in addition to a target power index, the algorithm takes a precision parameter $\epsilon$ and guarantees to output a weighted voting game of which the power index is $\epsilon$-close to it, on the precondition that there \emph{exists} an $\epsilon$-close weighted voting game with the property that the quota is not too skewed, in a particular sense. This is, to our knowledge, the only polynomial time algorithm for a power index voting game design problem that provides an approximation guarantee in any sense.

Closely related to our work are two papers that deal with the Chow parameters problem \cite{servediochow,dechow}. The results in their paper are stated in terms of boolean function theory and learning theory, but when translated to our setting, these papers can be seen to deal with approximation algorithms for a type of value that can be considered a power index: The Chow parameters of a given player in a given game is defined to be the total number of winning coalitions that the players is in. The authors present in these papers, as a main result, a polynomial time approximation scheme for computing the Chow parameters of a weighted voting game.

The problem of \emph{enumerating} the set of weighted voting games on a fixed number of players is, as we will see, closely related to the approach that we take for solving the weighted voting game design problem. This enumeration problem has been studied before in a paper by Kurz \cite{kurzminsum}, where the author uses integer programming techniques in order to enumerate all canonical weighted voting games on up to nine players. The author generates integer weighted representations for all of these games and classifies the games that do not have a unique minimum-sum integer weighted representation.

In \cite{krohnsudholter}, Krohn and Sudh\"olter study the enumeration of canonical linear games and a subclass thereof, using various order theoretic concepts. It does not directly address the problem of enumerating weighted voting games, although it discusses a correspondence between the $n$-player \emph{proper} weighted voting games and the $(n+1)$-player \emph{decisive} weighted canonical linear games.\footnote{A game is called proper if the complement of any winning coalition is losing. A game is called decisive if it is proper and the complement of any losing coalition is winning.} Because the class of canonical linear games is much bigger than the class of weighted voting games, their algorithms do not imply an efficient enumeration procedure for weighted voting games, as is one of our main contributions in the present work. However, there are some connections between our work and \cite{krohnsudholter}: their enumeration procedures work by exploiting graded posets, just like ours; although the posets
in question there are the subsets winning coalitions together with the set inclusion relation (for the case of decisive canonical linear games, they use a variant of this poset), and not on the subsets of minimal winning coalitions. Although their idea of using a graded poset corresponds with ours, it seems to us that our results cannot be connected to theirs in a stronger sense. Moreover, the proofs of the properties that we establish for the partially ordered set we propose here, use vastly different ideas, and crucially exploit weightedness.

Threshold functions \cite{hu1965threshold,murogathreshold} are of fundamental research interest in voting games, circuit complexity and neural networks. The problem of realizing Boolean threshold functions by neural networks has been extensively studied \cite{parberry1994circuit,siu1995discrete,freixas2008}, where upper and lower bounds are derived on the synaptic weights for such realization. The enumeration of threshold functions is closely related to the enumeration of weighted voting games (as threshold functions are essentially weighted voting games where negative weights are allowed.). 
The enumeration of threshold functions up to six variables has been done in \cite{murogasix}. Subsequently, in \cite{winderseven,murogaeight}, all threshold functions of respectively seven and eight variables were enumerated. Krohn and Sudh\"olter \cite{krohnsudholter} enumerated the canonical weighted voting games up to eight players, as well as the class of canonical linear games. Kurz \cite{kurzminsum} was the first to enumerate all nine player canonical weighted voting games, and Freixas and Molinero \cite{freixasunique} were the first to enumerate all nine player canonical linear games.

There exists some litature on enumeration of special subclasses of voting games as well: see \cite{freixastwotypes} for linear games with two desirability classes; \cite{kurzroof} for weighted voting games with one roof; and \cite{kurztypes} for linear games with certain special types of voters and few desirability classes.

Alon and Edelman observe that we need to know \textit{a priori} estimates of what power indices are achievable in simple games, in order to analyze the accuracy of these kinds of iterative algorithms, i.e., there is a need for information about the distribution of power indices in $[0,1]^n$. As a first step into solving this problem, they prove in \cite{alonedelman} a specific result for the case of the Banzhaf index for monotonic simple games.

Also, some applied work has been done on the design of voting games. In two papers, one by Laruelle and Widgr\'en \cite{laurellewidgren} and one by Sutter \cite{sutter}, the distribution of voting power in the European Union is analyzed and designed using iterative methods that resemble the algorithm of Aziz \cite{azizinverse}. Similar work was done by Leech for the EU \cite{leechdesigneu}, and for the IMF \cite{leechimf}.

Finally, a research direction that is related to our problem is that of studying minimal integer representation for a weighted voting game: Bounds on the maximum weight in such a representation provide us with a finite set of weighted representations to search through, as a means of solving our design problem. We explain this in greater detail in the next section. Some classical relevant bounds can be found in \cite{murogathreshold}, Section 9.3. See \cite{freixasmin,freixasunique} for some recent work in this direction.

\subsection{Outline}

The paper is divided into seven sections.

Section \ref{sec:preliminaries} introduces the required preliminary knowledge and defines some novel concepts. In particular it introduces cooperative games, with an emphasis on simple games (since our paper deals exclusively with simple games). We will also explain the notion of a power index, because a great part of what motivates the results of this paper has to do with a problem related to power indices. We give a definition of one of the most popular power indices (the Banzhaf index), and we briefly discuss algorithms for computing them, as well as the computational complexity of this problem.

We define in Section \ref{sec:problemstatement} the main problem of interest that we attempt to solve in this paper: the problem where we are given a target power index, and where we must find a game such that its power index is as close as possible to the given target power index. We explain that this specific problem is part of a more general family of problems that we call \text{voting game design problems}. 

Before directly trying to solve the problem introduced in Section \ref{sec:problemstatement}, we first discuss in Section \ref{sec:vgsynth} the problem of transforming various representations of games into each other. We give polynomial time algorithms for some of these problems, and also in some cases impossibility results regarding the existence of polynomial time algorithms. Also, in this section, we make some statements about the cardinality of certain classes of simple games.

Some of the results given in Section \ref{sec:vgsynth} are necessary for Section \ref{sec:PVGDsolving}, where we devise exact algorithms for the power index voting game design problem (our main problem of interest). We first explain a naive approach for the class of monotonic simple games, and after that, improve on this exponentially for the subclass of weighted voting games. We show how this improvement is possible by the existence of a certain partial order with some specific desirable properties. Next, we give various improvements to this algorithm by making use of the concepts of \emph{roof} and \emph{ceiling} coalitions.

After that, in Section \ref{sec:experiments} we will show various experimental results of a simple implementation of this exact algorithm. Because we can also use these algorithms as enumeration algorithms, we are able to provide some exact information about voting games, such as how many weighted voting games with a fixed number of minimal winning coalitions exist.

We conclude this paper in Section \ref{sec:conclusions}, with a discussion and some ideas for future work.

\section{Preliminaries}\label{sec:preliminaries}

In this section, we will discuss some required preliminary definitions and results.
Throughout this paper, we assume familiarity with big-O notation and analysis of algorithms.
In some parts of this paper, some basic knowledge of computational complexity theory is assumed as well, although these parts are not crucial for understanding the main results presented. We will not cover these topics in this section.

We use some order-theoretic notions throughout various sections of the paper. These are given in the following definition.
\begin{definition}[Partial order, (graded) poset, cover, rank function, least element]\label{def:posets}
For a set $S$, a \emph{partial order} $\preceq$ is a relation on $S$ that is \emph{reflexive}, so $\forall x\in S:x\preceq x$; \emph{antisymmetric}, so $\forall x,y\in S:((x\preceq y\wedge y\preceq x)\to x=y)$; and \emph{transitive}, so $\forall x,y,z\in S:((x\preceq y\wedge y\preceq z)\to x\preceq z)$. A \emph{partially ordered set} or \emph{poset} is a set $S$ equipped with a partial order $\preceq$, i.e., a pair $(S, \preceq)$.
A \emph{least element} of a poset $(S,\preceq)$ is an element $x \in S$ such that $x \preceq y$ for all $y \in S$.
A \emph{minimal element} of $(S,\preceq)$ is an element $x \in S$ such that $y \preceq x$ implies $y = x$ for all $y \in S$.
We say that $y \in S$ \emph{covers} $x \in S$ in $(S, \preceq)$ when $x \preceq y$ and there is no $z \in S$ such that $x \preceq z \preceq y$.
A poset $(S, \preceq)$ is \emph{graded} when there exists a \emph{rank function} $\rho : S \rightarrow \mathbb{N}$ such that: \bk{i.) $\rho$ is constant on all minimal elements of $(S, \preceq)$}. ii.) $\rho(x) \leq \rho(y)$ for all $x,y \in S$ such that $x \preceq y$. iii.) for any pair $x,y \in S$ it holds that if $y$ covers $x$ in $(S, \preceq)$, then $\rho(y) = \rho(x)+1$.
\end{definition}

The remainder of this section on preliminaries will be devoted to the theory of cooperative simple games.
A lot of the information in this section can be looked up in an introductory text on cooperative game theory, for example \cite{introcoopgametheory} or in Taylor and Zwicker's book on simple games \cite{simplegames}. We start with defining some essential terminology.
\begin{definition}[Cooperative (simple) games, grand coalition, characteristic function, monotonicity]\mbox{}
\begin{itemize}
\item A \emph{cooperative game} is a pair $(N,v)$, where $N$ is a finite set of \emph{players}; subsets of $N$ are called \emph{coalitions} and $v : 2^N \rightarrow \mathbb{R}_{\geq 0}$ is a function mapping coalitions to non-negative real numbers. Intuitively, $v$ describes how much collective payoff a coalition of players can gain when they cooperate.

\item $N$ is called the \emph{grand coalition}. $v$ is called the \emph{characteristic function} or \emph{gain function}.

\item A \emph{simple game} is a cooperative game $(N,v)$ where the codomain of $v$ is restricted to $\{0,1\}$. In this context, subsets of $N$ are referred to as \emph{winning coalitions} if $v(S) = 1$, and \emph{losing coalitions} otherwise (i.e., if $v(S) = 0$). (We purposefully do not exclude the game with only losing coalitions, as is customary. The reason is that including this game will make it more convenient later to show that a particular structure exists in a subclass of the simple games.)

\item A cooperative game $(N,v)$ is \emph{monotonic} if and only if $v(S) \leq v(T)$ for all pairs of coalitions $S,T \in 2^N$ that satisfy $S \subseteq T$.
\end{itemize}
\end{definition}
Note that in a large body of literature, the additional assumption is made that $v(\varnothing)$ in a cooperative game. For the sake of stating the results of our paper elegantly, we do not make this assumption.

We will often use simply the word \emph{game} to refer to a cooperative game.
A cooperative game $(N,v)$ will often be denoted by just $v$ when it is clear what the set of players is.
Later, we define various important additional classes of simple games. Since we will be working with these classes extensively, it is convenient to introduce the following notation in order to denote classes of games that are restricted to a fixed number of players $n$:
\begin{definition}
Let $\mathcal{G}$ be a class of games. Then we use $\mathcal{G}(n)$ to denote the class of games restricted to the set of players $\{1, \ldots, n\}$. In more formal language:
\begin{equation*}
\mathcal{G}(n) = \{ G : G \in \mathcal{G} \wedge G = (\{1,\ldots,n\}, v) \}.
\end{equation*}
\end{definition}
Throughout this paper, $n$ will always be the symbol we use to denote the number of players in a cooperative game.

The monotonic simple games are the games that we are concerned with in this paper.

\begin{definition}[The class of monotonic simple games]
We define $\mathcal{G}_{\mathsf{mon}}$ to be the class of all monotonic simple games.
\end{definition}

Some of the definitions in the remainder of this section are taken or adapted from \cite{harisazizinfluence} and \cite{simplegames}. Next, we turn to some syntactic definitions of certain classes of simple games. There are various important ways to represent simple games:
\begin{definition}[Representations of simple games]
Suppose that $(N,v)$ is a simple game. Let $W = \{S : S \in 2^N \wedge v(S) = 1\}$ and $L = \{S : S \in 2^N \wedge v(S) = 0\}$ be its sets of respectively \emph{losing coalitions} and \emph{winning coalitions}. Define $W_{\min} = \{S \in W : (\forall i \in S) v(S \setminus \{i\}) = 0\}$ and $L_{\max}= \{S \in L : (\forall i \in N \setminus S) v(S \cup \{i\}) = 1\}$ as their respective sets of \emph{minimal winning coalitions} and \emph{maximal losing coalitions}. We can describe a simple game in the following forms:
\begin{description}
 \item[Winning coalition form] $(N, W)$ is called the \textit{winning coalition form} of $(N,v)$.
 \item[Losing coalition form] $(N, L)$ is called the \textit{losing coalition form} of $(N,v)$.
 \item[Minimal winning coalition form] If $(N,v)$ is monotonic, then $(N, W_{\min})$ is the \textit{minimal winning coalition form} of $(N,v)$. Observe that $W_{\min}$ fully describes $v$ if and only if $(N,v)$ is monotonic.
 \item[Maximal losing coalition form] If $(N,v)$ is monotonic, then $(N, L_{\max})$ is the \textit{maximal losing coalition form} of $(N,v)$. Observe that $L_{\max}$ fully describes $v$ if and only if $(N,v)$ is monotonic.
 \item[Weighted form] If there exists a quota $q \in \mathbb{R}_{\geq 0}$ and a weight $w_i \in \mathbb{R}_{\geq 0}$ for each player $i \in N$, such that for each coalition $S \in 2^N$ it holds that $v(S) = 1 \Leftrightarrow \sum_{i \in S} w_i \geq q$, then the vector $w = (q, w_1, \ldots,w_n)$, also written as $[q; w_1, \ldots, w_n]$, is called a \text{weighted form} of $(N,v)$. Observe that every game that has a weighted form is also monotonic.
\end{description}
\end{definition}

Games that have a weighted form are of our main interest and have a special name:
\begin{definition}[Weighted voting games]
If a monotonic simple game has a weighted form, then it is called a \emph{weighted voting game}. The class of all weighted voting games is denoted by $\mathcal{G}_{\mathsf{wvg}}$.
\end{definition}
It is well known that the class of weighted voting games is strictly contained in the class of monotonic simple games: examples of monotonic simple games that are not weighted are numerous and easily constructed. Later, in Section \ref{cardinalities}, we discuss the cardinalities of these classes with respect to $n$.

A weighted voting game is an important type of simple game because it has a compact representation. Also, weighted voting games are important because they are used in a lot of practical situations, i.e., in a lot of real-life decision making protocols, for example: elections, politics and stockholder companies.
An important property of weighted voting games that we will use, is that a weighted representation of such a game is invariant to scaling:
\begin{proposition}\label{wvginvmult}
Let $G \in \mathcal{G}_{\mathsf{wvg}}(n)$ be a weighted voting game, and let $\ell = [q; w_1, \ldots, w_n]$ be a weighted representation for $G$. For every $\lambda \in \mathbb{R}^+$, we have that $\ell' = [\lambda q; \lambda w_1, \ldots, \lambda w_n]$ is a weighted representation for $G$.
\end{proposition}
\begin{proof}
For any coalition $C \subseteq N$ such that $w_\ell(C) < q$:
\begin{equation*}
w_{\ell'}(C) = \sum_{i \in C} \lambda w_i = \lambda \sum_{i \in C} w_i = \lambda w_\ell(C) < \lambda q,
\end{equation*}
and for any coalition $C \subseteq N$ such that $w_\ell(C) \geq q$:
\begin{equation*}
w_{\ell'}(C) = \sum_{i \in C} \lambda w_i = \lambda \sum_{i \in C} w_i = \lambda w_\ell(C) \geq \lambda q.
\end{equation*}
\end{proof}

We will be using the following notational abuse in the remainder of this paper: Whenever we are discussing a weighted voting game $G$ with players $N = \{1,\ldots, n\}$ and weighted form $[q; w_1, \ldots, w_n]$, we use $w(S)$ as a shorthand for $\sum_{i \in S} w_i$ for any subset $S$ of $N$.

We next turn our attention to the topic of influence and power in monotonic simple games. For a monotonic simple game, it is possible to define a relation called the \textit{desirability relation} on the players (see \cite{isbelldesirability}):
\begin{definition}[Desirability relation]
For a monotonic simple game $(N,v)$, the \emph{desirability relation} $\succeq_v$ is defined by:
\begin{itemize}
\item For any $i,j \in N:$ if $\forall S \subseteq N \setminus \{i,j\} : v(S \cup \{i\}) \geq v(S \cup \{j\})$, then $i \succeq_v j$. In this case we say that $i$ is \textit{more desirable} than $j$.
\item For any $i,j \in N:$ if $\forall S \subseteq N \setminus \{i,j\} : v(S \cup \{i\}) = v(S \cup \{j\})$, then $i \sim_v j$. In this case we say that $i$ and $j$ are \textit{equally desirable}.
\item For any $i,j \in N:$ if $\forall S \subseteq N \setminus \{i,j\} : v(S \cup \{i\}) \leq v(S \cup \{j\})$, then $i \preceq_v j$. In this case we say that $i$ is \textit{less desirable} than $j$.
\item For any $i,j \in N:$ if $i \succeq_v j$ and not $i \sim_v j$, then $i \succ_v j$. In this case we say that $i$ is \textit{strictly more desirable} than $j$.
\item For any $i,j \in N:$ if $i \preceq_v j$ and not $i \sim_v j$, then $i \prec_v j$. In this case we say that $i$ is \textit{strictly less desirable} than $j$.
\end{itemize}
Moreover, if neither $i \succeq_v j$ nor $j \succeq_v i$ holds for some $i,j \in N$, then we say that $i$ and $j$ are \textit{incomparable}.

In cases that it is clear which game is meant, we drop the subscript and write $\preceq, \prec, \succeq, \succ, \sim$ instead of $\preceq_v, \prec_v, \succeq_v, \succ_v, \sim_v$.
\end{definition}
There exist other notions of desirability, for which different properties hold \cite{desirability}.
In the context of other desirability relations, the desirability relation that we have defined here is refered to as the \emph{individual desirability relation}. Since this is the only desirability relation that we will use in this paper, we will refer to it as simply the \emph{desirability relation}.

Using the notion of this desirability relation, it is now possible to define the class of \textit{linear games}.
\begin{definition}[Linear game]
A simple game $(N,v)$ is a \textit{linear game} if and only if it is monotonic, and in $(N,v)$ no pair of players in $N$ is incomparable with respect to $\preceq$. Thus, for a linear game $(N,v)$, $\preceq$ is a total preorder on $N$. We denote the class of linear games by $\mathcal{G}_{\mathsf{lin}}$.
\end{definition}

It is straightforward to see that all weighted voting games are linear: let $(N,v)$ be a weighted voting game where $N = \{1, \ldots, n\}$, and let $[q; w_1, \ldots, w_n]$ be a weighted form of $(N,v)$. Then it holds that $i \preceq j$ when $w_i \leq w_j$. Hence, every pair of players is comparable with respect to $\preceq$.

In fact, the following sequence of strict containments holds: $\mathcal{G}_{\mathsf{wvg}} \subset \mathcal{G}_{\mathsf{lin}} \subset \mathcal{G}_{\mathsf{mon}}$. This brings us to the definition of two special classes of games that will be convenient for use in subsequent sections.
\begin{definition}[Canonical weighted voting games \& canonical linear games]
A linear game $(N,v)$ is a \textit{canonical linear game} whenever $N = \{1, \ldots, n\}$ for some $n \in \mathbb{N}_{> 0}$, and the desirability relation $\succeq$ satisfies $1 \succeq 2 \succeq \cdots \succeq n$. When $G$ is also weighted, then $G$ is a \textit{canonical weighted voting game}. The class of canonical linear games is denoted by $\mathcal{G}_{\mathsf{clin}}$, and the class of canonical weighted voting games is denoted by $\mathcal{G}_{\mathsf{cwvg}}$.
\end{definition}
Note that a canonical weighted voting game always has a weighted representation that is non-increasing.

It is now time to introduce two special ways of representing canonical linear games.
\begin{definition}[Left-shift \& right-shift]\label{leftshiftrightshift}
Let $N$ be the set of players $\{1, \ldots, n\}$ and let $S$ be any subset of $N$. A coalition $S' \subseteq N$ is a \textit{direct left-shift} of $S$ whenever 
there exists an $i\in S$ and $i-1 \not\in S$ with $2 \leq i \leq n$
such that $S' = (S \setminus \{i\}) \cup \{i-1\}$. A coalition $S' \subseteq N$ is a \textit{left-shift} of $S$ whenever for some $k \geq 1$ there exists a sequence $(S_1, \ldots, S_k) \in (2^n)^k$, such that
\begin{itemize}
\item $S_1 = S$,
\item $S_k = S'$,
\item for all $i$ with $1 \leq i < k$, we have that $S_{i+1}$ is a direct left-shift of $S_{i}$.
\end{itemize}
The definitions of \textit{direct right-shift} and \textit{right-shift} are obtained when we replace in the above definition $i-1$ with $i+1$ and $i+1$ with $i-1$.
\end{definition}

For example, coalition $\{1,3,5\}$ is a \emph{direct} left-shift of coalition $\{1,4,5\}$, and coalition $\{1,2,5\}$ is a left-shift of $\{1,4,5\}$.

The notions of left-shift and right-shift make sense for canonical linear games and canonical weighted voting games:
Because of the specific desirability order that holds in canonical linear games, a left-shift of a winning coalition is always winning in such a game, and a right-shift of a losing coalition is always losing in such a game. This allows us to represent a canonical linear game in one of the following two forms.
\begin{definition}[Roof/ceiling coalition/form]
Let $(N = \{1, \ldots, n\}, v)$ be a canonical linear game. Also, let $W_{\min}$ be $(N,v)$'s list of minimal winning coalitions and let $L_{\max}$ be $(N,v)$'s list of maximal losing coalitions. A minimal winning coalition $S \in W_{\min}$ is a \textit{roof coalition} whenever every right-shift of $S$ is losing. Let $W_{\mathsf{roof}}$ denote the set of all roof coalitions of $G$. The pair $(N, W_{\mathsf{roof}})$ is called the \textit{roof form} of $G$.
A maximal losing coalition $S \in L_{\max}$ is a \textit{ceiling coalition} whenever every left-shift of $S$ is winning. Let $W_{\mathsf{ceil}}$ denote the set of all ceiling coalitions of $G$. The pair $(N, W_{\mathsf{ceil}})$ is called the \textit{ceiling form} of $G$.
\end{definition}
The terminology (``roof'' and ``ceiling'') is taken from \cite{hopskipjump}, although they have also been called \textit{shift-minimal winning coalitions} and \emph{shift-maximal losing coalitions} \cite{simplegames}.

Because we will be discussing simple games from a computational perspective, we next introduce the concept of \textit{representation languages} for simple games.

\subsection{Representation languages}
We have introduced several ways of representing simple games: by the sets of winning and losing coalitions; by the sets of minimal winning coalitions and maximal losing coalitions; by the sets of roof coalitions and ceiling coalitions; and by their weighted representation.

We now make precise the notion of \emph{representing a simple game} by turning these methods representing simple games into languages: sets of strings, such that the strings are a description of a game according to one of the methods in the list above.

Prior to defining these languages, we need a way of describing coalitions. Coalitions can be described using their \textit{characteristic vector}.
\begin{definition}\label{characteristicvector}
Let $N = \{1, \ldots, n\}$ be a set of $n$ players. The \textit{characteristic vector} $\vec{\chi}(S)$ of a coalition $S \subseteq N$ is the vector $(\chi(1,S), \ldots, \chi(n,S))$ where
\begin{equation*}
\chi(i,S) = \begin{cases}1 \text{ if } i \in S \\ 0 \text{ otherwise. }\end{cases}
\end{equation*}
\end{definition}

A characteristic vector of a coalition in a game of $n$ players is described by $n$ bits.

\begin{definition}[Representation Languages]
We define the following \emph{representation languages} to represent simple games.
\begin{itemize}
\item $\mathcal{L}_{W}$. Strings $\ell \in L_{W}$ are lists of characteristic vectors of coalitions. The string $\ell$ represents a simple game $G$ if and only if the set of coalitions that $\ell$ describes is precisely the set of coalitions that are winning in $G$.
\item The languages $\mathcal{L}_{W, \min}$, $\mathcal{L}_{L}$, $\mathcal{L}_{L, \max}$, $\mathcal{L}_{\mathsf{roof}}$, and $\mathcal{L}_{\mathsf{ceil}}$ are defined in the obvious analogous fashion.
\item $\mathcal{L}_{\mathsf{weights}}$. Strings $\ell \in L_{\mathsf{weights}}$ are lists of numbers $\langle q, w_1, \ldots, w_n \rangle$. The string $\ell$ represents the simple game $G$ if and only if $G$ is a weighted voting game with weighted form $[q; w_1, \ldots, w_n]$.
\end{itemize}
\end{definition}

We will use the following convention: For a representation language $\mathcal{L}$, we denote with $\mathcal{L}(n)$ the set of strings in the language $\mathcal{L}$ that represent games of $n$ players. Also, let $\ell$ be a string from a representation language $\mathcal{L}(n)$. Then we write $G_{\ell}$ to denote the simple game on players $\{1,\ldots,n\}$ that is represented by $\ell$.

\begin{definition}
We say that a class of games $\mathcal{G}$ \textit{is defined by} a language $\mathcal{L}$ if and only if $\forall \ell \in \mathcal{L} : \exists G \in \mathcal{G} : G_{\ell} = G$ and vice versa $\forall G \in \mathcal{G} : \exists \ell \in \mathcal{L} : G_{\ell} = G$.
\end{definition}
Using the above definition, we see that
\begin{itemize}
\item $\mathcal{G}_{\mathsf{sim}}$ is defined by both $\mathcal{L}_{W}$ and $\mathcal{L}_{L}$;
\item $\mathcal{G}_{\mathsf{mon}}$ is defined by both $\mathcal{L}_{W, \min}$ and $\mathcal{L}_{L, \max}$;
\item $\mathcal{G}_{\mathsf{lin}}$ is defined by both $\mathcal{L}_{\mathsf{roof}}$ and $\mathcal{L}_{\mathsf{ceil}}$;
\item $\mathcal{G}_{\mathsf{wvg}}$ is defined by $\mathcal{L}_{\mathsf{weights}}$.
\end{itemize}

\subsection{Power indices}\label{powerindices}

\emph{Power indices} can be used to measure the amount of influence that a player has in a monotonic simple game. Power indices were originally introduced because it was observed that in weighted voting games, the weight of a player is not directly proportional to the influence he has in the game. This is easy to see through the following trivial example weighted voting game:
\begin{equation*}
[1000; 997, 1, 1, 1].
\end{equation*}
Here, each player is in only one winning coalition: the grand coalition. All players are required to be present in this coalition for it to be winning, and can therefore be said to have the same influence, despite the fact that there is a huge difference between the weights of the first vs.\ the other three players.

Many proposals have been put forward to answer the question of what constitutes a good definition of power in a voting game. These answers are in the form of \textit{power indices}, which are mathematical formulations for values that try to describe the `true' influence a player has in a weighted voting game. We refer the reader to \cite{powerwebsite} for an excellent WWW information resource on power indices.

Power indices try to measure a player's \textit{a priori} power in a voting game. That is, they attempt to objectively measure the influence a player has on the outcome of a voting game, without having any statistical information on which coalitions are likely to form due to the preferences of the players. To do this, we cannot avoid making certain assumptions, but we let these assumptions be as neutral as possible. For example, in the Banzhaf index we describe below, the assumption is that each coalition will form with equal probability.

While the need for power indices originally arose from studying weighted voting games, all of the power indices that have been devised up till now also make sense for (and are also well-defined for) simple games. So, for any simple coalitional game, we can use a power index as a measure of a player's a priori power in it.

In this paper we use the \emph{normalized Banzhaf index} as our power index of choice. This power index is used in the experiments discussed in Section \ref{sec:experiments}. However, we will see that for the theoretical part of our work, the particular choice of power index is irrelevant.

\begin{definition}[normalized Banzhaf index \& raw Banzhaf index]
For a monotonic simple game $(N = \{1, \ldots, n\}, v)$, the \textit{normalized Banzhaf index} of $(N,v)$ is defined as $\beta = (\beta_1, \ldots, \beta_n)$, where for $1 \leq i \leq n$,
\begin{equation*}
\beta_i = \frac{\beta_i'}{\sum_{j=1}^n \beta_j'},
\end{equation*}
and
\begin{equation}\label{rawbanzhaf}
\beta_i' = |\{S \subseteq N \setminus \{i\} : v(S) = 0 \wedge v(S \cup \{i\}) = 1\}|.
\end{equation}
Here, $\beta_i'$ is called the \textit{raw Banzhaf index of player $i$}.
\end{definition}
Note that the Banzhaf index of an $n$-player simple game is always a member of the unit simplex of $\mathbb{R}^n$, i.e., the set $\left\{x \in \mathbb{R}^n : \sum_{i=1}^n x_i = 1\right\}$.

The problem of computing power indices, and its associated computational complexity, has been widely studied (e.g., in \cite{matsui01npcompleteness,matsuisurvey,hemaspaandracomparison,bachrachnetworkflow,bachrachconnectivity,algababilbao,woeginger,bilbaofernandez,unoimprovement,Bachrach:2008:AAMASd,leechapproximation,leechsurvey,powerindexsurvey}). For a survey of complexity results, exact algorithms and approximation algorithms for computing power indices, see \cite{powerindexsurvey}. In general, computing power indices is a hard task, and the case of the normalized Banzhaf index is no exception: Computation of the raw Banzhaf index is known to be $\mathsf{\#P}$-complete \cite{prasadkelly1990}, and the fastest known exponential time algorihm for computing the Banzhaf index is due to Klinz and Woeginger \cite{woeginger}. It achieves a runtime in $O((\sqrt{2})^n \cdot n^2)$.

\section{The problem statement}\label{sec:problemstatement}
In this section, we will introduce the problem that we call the \textit{voting game design} problem:
the problem of finding a simple game that satisfies a given requirement (or set of requirements) as well as possible. We will focus on the problem of finding games in which the power index of the game is as close as possible to a given target power index.

We define a voting game design problem as an optimization problem where we are given three parameters $f$, $\mathcal{G}$, and $\mathcal{L}$.
In such a voting game design problem we must minimize some function $f : \mathcal{G} \rightarrow \mathbb{R}_{\geq 0}$, with $\mathcal{G}$ being some class of simple games. $\mathcal{L}$ is a representation language for $\mathcal{G}$. We require the game that we output to be in the language $\mathcal{L}$.

\begin{definition}[$(f, \mathcal{G}, \mathcal{L})$-voting game design ($(f, \mathcal{G}, \mathcal{L})$-VGD)]
Let $\mathcal{G}$ be a class of simple games, let $\mathcal{L}$ be a representation language for $\mathcal{G}$, and let $f : \mathcal{G} \rightarrow \mathbb{R}^+  \cup \{0\}$ be a function.
The \textit{$(f, \mathcal{G}, \mathcal{L})$-voting game design problem} (or \textit{$(f, \mathcal{G}, \mathcal{L})$-VGD}) is
the problem of finding an $\ell \in \mathcal{L}$ such that $G_{\ell} \in \mathcal{G}$ and $f(G_{\ell})$ is minimized.
\end{definition}
Hence, $f$ can be seen as a function indicating the \emph{error}, or the \emph{distance} from the game that we are ideally looking for.
By imposing restrictions on the choice of $f$, and by fixing $\mathcal{G}$ and $\mathcal{L}$, we can obtain various interesting optimization problems.
The cases that we will focus on will be those where $f$ is a function that returns the distance of a game's power index from a certain \emph{target} power index.

\begin{definition}[($g, \mathcal{G}, \mathcal{L}$)-power index voting game design (($g, \mathcal{G}, \mathcal{L}$)-PVGD)]\label{pvgdproblemdef}
Suppose $\mathcal{G}$ is a class of games, and $\mathcal{L}$ is a representation language for a class of games. Furthermore, suppose $g : \mathcal{G} \rightarrow \mathbb{R}^n$ is a function that returns a type of power index (e.g., the normalized Banzhaf index) for games in $\mathcal{G}$.
Then, the \textit{($g, \mathcal{G}, \mathcal{L}$)-power index voting game design problem} (or \textit{($g, \mathcal{G}, \mathcal{L}$)-PVGD}) is the
($f, \mathcal{G}, \mathcal{L}$)-VGD problem with $f$ restricted to those functions for which there exists a vector $(p_1, \ldots, p_n)$
such that for each $G \in \mathcal{G}$,
\begin{equation*}
f(G) = \sqrt{\sum_{i=1}^n (g(G)_i - p_i)^2}.
\end{equation*}
\end{definition}
In words, in a ($g, \mathcal{G}, \mathcal{L}$)-PVGD problem we must find a voting game in the class $\mathcal{G}$ that is as close as possible to a given target power index $(p_1, \ldots, p_n)$ according to power index function $g$ and error function $f$. In this paper we measure the error by means of the Euclidean distance in $\mathbb{R}^n$ between the power index of the game and the target power index.
We made this particular choice of $f$ because from intuition it seems like a reasonable error function. In principle, we could also choose $f$ differently. For example, we could take for $f$ any other norm on $\mathbb{R}^n$. For our purpose, the precise choice of $f$ does not really matter, as long as the error function is not hard to compute, given $g$.

We can analyze this problem for various power index functions, classes of games, and representation languages. So, an instance of such a problem is then represented by only a vector $(p_1, \ldots, p_n)$, representing a target power index.

We will focus in this paper on the problem ($\beta, \mathcal{G}_{\mathsf{wvg}}, \mathcal{L}_{\mathsf{weights}}$)-PVGD, i.e., the problems of finding a weighted voting game in weighted representation, that is as close as possible to a certain target (normalized) Banzhaf index.

\section{Voting game synthesis}\label{sec:vgsynth}
The method that we will propose for solving the power index voting game design problem involves transforming between different representations for classes of simple games. In this section we will give an overview of transforming representations of simple games into each other. We call these problems \textit{voting game synthesis problems}, inspired by the term \textit{threshold synthesis} used in \cite{hopskipjump} for finding a weight vector for a so-called \textit{threshold function}, to be defined later in this section.

In Section \ref{cardinalities}, we first find out what we can say about the cardinalities of various classes of voting games.
We state the synthesis problem formally in Section \ref{formalsynthesisproblem}. In Section \ref{solvingsynthesisproblems} we will look at how to solve it.

\subsection{On cardinalities of classes of simple games}\label{cardinalities}
In some variants of the voting game synthesis problem, we want to transform a simple game into a specific representation language that defines only a subclass of the class of games that is defined by the input-language.

It is interesting to know what fraction of a class of games is synthesizable in which language, i.e., we are interested in the cardinalities of all of these classes of simple games. This is an interesting question in its own right, but we also require it in order to analyze the algorithms for the voting game design problems that we will present in the next sections.

First we will discuss the number of monotonic simple games and linear games on $n$ players. After that, we will also look at the number of weighted voting games on $n$ players.

\subsubsection{The number of monotonic simple games and linear games}\label{cardinalities1}
Let us start off with the cardinality of the class of monotonic simple games of $n$ players: $\mathcal{G}_{\mathsf{mon}}(n)$. This class is defined by language  $\mathcal{L}_{W,\min}$, i.e., each game in this class can be described by a set of minimal winning coalitions (MWCs), and for each possible set of MWCs $W_{\min}$ there is a monotonic simple game $G$ such that the MWCs of $G$ are precisely $W_{\min}$. We see therefore that the number of monotonic simple games on $n$ players is equal to the number of families of MWCs on $n$ players. In a family of MWCs, there are no two coalitions $S$ and $S'$ such that $S'$ is a superset of $S$. In other words: all elements in a family of MWCs are pairwise incomparable with respect to $\subseteq$, or: A family of MWCs on the set of players $N = \{1, \ldots, n\}$ is an \emph{antichain} in the poset $(2^N, \subseteq)$. Hence, the number of antichains in this poset is equal to $|\mathcal{G}_{\mathsf{mon}}(n)|$. Counting the number of antichains in this poset is a famous
open problem in combinatorics, known as \textit{Dedekind's problem} and $|\mathcal{G}_{\mathsf{mon}}(n)|$ is therefore also referred to as the \textit{$n$th Dedekind number} $D_n$. Dedekind's problem was first stated in \cite{dedekindfirst}. To the best of our knowledge, exact values for $D_n$ are known only up to $n = 8$. We will return to the discussion of the Dedekind number in Section \ref{mongamedesign}, where we will also mention some known upper and lower bounds for it. For now, let us simply say that $D_n$ grows rather quickly in $n$: as $n$ gets larger, $D_n$ increases exponentially.

For linear games, we know only of the following lower bound on the number of canonical linear games. The prove that we give here is from \cite{hopskipjump}:
\begin{theorem}\label{lineargameslowerbound}
For large enough $n$,
\begin{equation*}
|\mathcal{G}_{\mathsf{clin}}(n)| \geq 2^{(\sqrt{\frac{2}{3}\pi}2^n) / (n \sqrt n)}.
\end{equation*}
\end{theorem}
\begin{proof}
First observe that $|\mathcal{G}_{\mathsf{clin}}(n)|$ is equal to the number of antichains in the poset $(2^N, \preceq_{\mathsf{ssrs}})$, where $\preceq_{\mathsf{ssrs}}$ is defined as follows: for two coalitions $S \subseteq N$ and $S' \subseteq N$, we have $S \preceq_{\mathsf{ssrs}} S'$ if and only if $S$ is a superset of a right-shift of $S'$. It can be seen that $(2^N, \preceq_{\mathsf{ssrs}})$ is a graded poset, with the following rank function $r$:
\begin{eqnarray*}
r & : & 2^N \rightarrow \mathbb{N} \\
 & & S \mapsto \sum_{i \in S} n-i+1.
\end{eqnarray*}
A set of points of the same rank is an antichain in $(2^N, \preceq_{\mathsf{ssrs}})$. Let $A_k$ denote the set of points of rank $k$. $k$ is at most $\frac{n(n+1)}{2}$. For each coalition $S$ in $A_k$, its complement $N \setminus S$ is in $A_{n(n+1)/2 - k}$; therefore $|A_k| = |A_{n(n+1)/2 - k}|$. It is shown in \cite{stanley} that the sequence $(|A_1|, \ldots, |A_{n(n+1)/2}|)$ is unimodal, i.e., first non-increasing, then non-decreasing. By this fact and the fact that $|A_k| = |A_{n(n+1)/2 - k}|$, it must be the case that the largest antichain is $|A_{n(n+1)/4}|$. $|A_{n(n+1)/4}|$ is equal to the number of points $(x_1, \ldots, x_n)$ satisfying $x_1 + 2x_2 + \cdots + nx_n = \frac{n(n+1)}{4}$, and this number of points is equal to the middle coefficient of the polynomial $(1 + q)(1 + q^2)\cdots(1+q^n)$. It is shown in \cite{odlyzkorichmond} that this middle coefficient is asymptotically equal to
\begin{equation*}
\frac{\sqrt{\frac{2}{3}\pi}2^n}{n \sqrt n}
\end{equation*}
Since every subset of an antichain is also an antichain, there must be more than
\begin{equation*}
2^{(\sqrt{\frac{2}{3}\pi}2^n) / (n \sqrt n)}
\end{equation*}
antichains in $(2^N, \preceq_{\mathsf{ssrs}})$.
\end{proof}

\subsubsection{The number of weighted voting games}\label{cardinalities2}
To our knowledge, the existing game theory literature does not provide us with any general insights in the number of weighted voting games on $n$ players. Fortunately there is a closely related field of research, called \textit{threshold logic} (see for example \cite{murogathreshold}), that has some relevant results.
\begin{definition}[Boolean threshold function, realization, $\mathbf{LT}$]
Let $f$ be a boolean function on $n$ boolean variables. $f$ is a \textit{(boolean) threshold function} when there exists a weight vector of real numbers
$r = (r_0, r_1, \ldots r_n) \in \mathbb{R}^{n+1}$ such that $r_1x_1 + \cdots + r_nx_n \geq r_0$ if and only if $f(x_1, \ldots, x_n) = 1$.
We say that $r$ \textit{realizes} $f$. We denote the set of threshold functions of $n$ variables $\{x_1, \ldots, x_n\}$ by $\mathbf{LT}(n)$.\footnote{``LT'' stands for ``Linear Threshold function''.}
\end{definition}
Threshold functions resemble weighted voting games, except for that we talk about \textit{boolean variables} instead of \textit{players} now.
Also, an important difference between threshold functions and weighted voting games is that $r_0, r_1, \ldots, r_n$ are allowed to be negative for threshold functions, whereas $q, w_1, \ldots, w_n$, must be non-negative in weighted voting games.

\cite{zunicenumeration} gives an upper bound on the number of threshold functions of $n$ variables $|\mathbf{LT}(n)|$:
\begin{equation*}
|\mathbf{LT}(n)| \leq 2^{n^2-n+1}.
\end{equation*}
Also, the following asymptotic lower bound is known, as shown in \cite{zuevasymptotics}: For large enough $n$, we have
\begin{equation}\label{lowerboundthresholdfunctions}
|\mathbf{LT}(n)| \geq 2^{n^2(1-\frac{10}{\log n})}.
\end{equation}

From these bounds, we can deduce some easy upper and lower bounds for $|\mathcal{G}_{\mathsf{wvg}}|$.

First we observe the following property of the set of threshold functions on $n$ variables.
Let $\mathbf{LT}^+(n)$ be the set of \textit{non-negative threshold functions} of variables $\{x_1, \ldots, x_n\}$: threshold functions $f \in \mathbf{LT}(n)$ for which there exists a \textit{non-negative weight vector} $r$ that realizes $f$
It is then not hard to see that there is an obvious one-to-one correspondence between the games in $\mathcal{G}_{\mathsf{wvg}}(n)$ and the threshold functions in $\mathbf{LT}^+(n)$, so $|\mathcal{G}_{\mathsf{wvg}}(n)| = |\mathbf{LT}^+(n)|$.
An easy upper bound then follows:
\begin{corollary}
For all $n$, $|\mathcal{G}_{\mathsf{wvg}}(n)| \leq 2^{n^2-n+1}$.
\end{corollary}

We will proceed by obtaining a lower bound on the number of weighted voting games.
\begin{corollary}\label{no-wvgs}
For large enough $n$, it holds that
\begin{equation*}
|\mathcal{G}_{\mathsf{wvg}}(n)| \leq 2^{n^2(1-\frac{10}{\log n}) - n - 1}
\end{equation*}
\end{corollary}
\begin{proof}
Let $f$ be a non-negative threshold function and let $r$ be a non-negative weight vector that realizes $f$.
There are $2^{n+1}$ possible ways to negate the elements of $r$, so there are at most
$2^{n+1} - 1$ threshold functions $f' \in \mathbf{LT}(n) \setminus \mathbf{LT}^+(n)$ such that $f'$ has a realization that is obtained by negating some of the elements of $r$. From this, it follows that $|\mathbf{LT}^+(n)| \geq \frac{|\mathbf{LT}(n)|}{2^{n+1}}$, and thus also $|\mathcal{G}_{\mathsf{wvg}}(n)| \geq |\frac{\mathbf{LT}(n)}{2^{n+1}}|$. Now by using (\ref{lowerboundthresholdfunctions}) we get $|\mathcal{G}_{\mathsf{wvg}}(n)| \geq \frac{2^{n^2(1-\frac{10}{\log n})}}{2^{n+1}} = 2^{n^2(1-\frac{10}{\log n}) - n - 1}$.
\end{proof}
We have obtained this lower bound on the number of weighted voting games by upper-bounding the factor, say $k$, by which the number of threshold functions is larger than the number of non-negative threshold functions. If we could find the value of $k$ exactly, or at least lower-bound $k$, then we would also be able to sharpen the upper bound on the number of weighted voting games.

Our next question is: what about the canonical case, $\mathcal{G}_{\mathsf{cwvg}}(n)$? $\mathcal{G}_{\mathsf{cwvg}}(n)$ is a subset of $\mathcal{G}_{\mathsf{wvg}}(n)$, and for each non-canonical weighted voting game there exists a permutation of the players that makes it a canonical one. Since there are $n!$ possible permutations, it must be that $|\mathcal{G}_{\mathsf{cwvg}}(n)| \geq \frac{|\mathcal{G}_{\mathsf{wvg}}(n)|}{n!}$, and thus we obtain that
\begin{equation}\label{no-cwvgs}
|\mathcal{G}_{\mathsf{cwvg}}(n)| \geq \frac{2^{n^2(1-\frac{10}{\log n}) - n - 1}}{n!}
\end{equation}
for large enough $n$.

\subsection{The synthesis problem for simple games}\label{formalsynthesisproblem}
In a \textit{voting game synthesis problem}, we are interested in transforming a given simple game from one representation language into another representation language.

\begin{definition}[Voting game synthesis (VGS) problem]
Let $\mathcal{L}_1$ and $\mathcal{L}_2$ be two representation languages for (possibly) distinct classes of simple games. Let $f_{\mathcal{L}_1 \rightarrow \mathcal{L}_2} : \mathcal{L}_1 \rightarrow \mathcal{L}_2 \cup \{\mathsf{no}\}$ be the function that, on input $\ell$,
\begin{itemize}
\item outputs $\mathsf{no}$ when $G_{\ell}$ is not in the class of games defined by $\mathcal{L}_2$,
\item otherwise maps a string $\ell \in \mathcal{L}_1$ to a string $\ell' \in \mathcal{L}_2$ such that $G_{\ell} = G_{\ell'}$ .
\end{itemize}
In the \textit{($\mathcal{L}_1, \mathcal{L}_2)$-voting game synthesis problem}, or \textit{($\mathcal{L}_1, \mathcal{L}_2)$-VGS problem}, we are given a string $\ell \in \mathcal{L}_1$ and we must compute $f_{\mathcal{L}_1 \rightarrow \mathcal{L}_2}(\ell)$.
\end{definition}

\subsection{Algorithms for voting game synthesis}\label{solvingsynthesisproblems}
In this section we will discuss algorithms and hardness results for various VGS problems.
In Sections \ref{weightedvgs}, \ref{roofvgs} and \ref{othervgs} we will consider respectively the problems of
\begin{itemize}
\item transforming games into weighted representation ($\mathcal{L}_{\mathsf{weights}}$);
\item transforming games into roof- or ceiling-representation ($\mathcal{L}_{\mathsf{roof}}, \mathcal{L}_{\mathsf{ceil}}$);
\item transforming games into the languages $\mathcal{L}_W, \mathcal{L}_{W,\min}, \mathcal{L}_L, \mathcal{L}_{L,\max}$.
\end{itemize}

\subsubsection{Synthesizing weighted representations}\label{weightedvgs}
For our approach to solving the power index voting game design problem for weighted voting games, which we will present in Section \ref{sec:PVGDsolving}, it is of central importance that the problem ($\mathcal{L}_{W,\min}, \mathcal{L}_{\mathsf{weights}}$)-VGS has a polynomial time algorithm. This is a non-trivial result and was first stated in \cite{hopskipjump} by Peled and Simeone. In \cite{hopskipjump}, the problem is stated in terms of \emph{set-covering problems}. Because this algorithm is central to our approach for solving the PVGD-problem, we will here restate the algorithm in terms of simple games, and we will give a proof of its correctness and polynomial time complexity.

In order to state the algorithm, we first introduce a new total order on the set of coalitions $2^N$ of a set of players $N = \{1, \ldots, n\}$.
\begin{definition}[Positional representation]
Let $S \subseteq N = \{1, \ldots, n\}$ be a coalition.
The \textit{$i$th position} $p(i,S)$ of $S$ is defined to be the player $a$ in $S$ such that $|\{1, \ldots, a\} \cap N| = i$.
The \textit{positional representation} of $S$, $\mathsf{pr}(S)$, is defined as the $n$-dimensional vector $(p'(1,S), \ldots, p'(n,S))$ where
\begin{equation*}
p'(i,S) = \begin{cases} 0 \text{ if } |S| < i, \\ p(i,S) \text{ otherwise. } \end{cases}
\end{equation*}
for all $i$ with $1 \leq i \leq n$.
\end{definition}
As an example: if we have $N = \{1, \ldots, 5\}$ and $S = \{1,4,5\}$, then $\mathsf{pr}(S) = (1,4,5,0,0)$.
\begin{definition}[PR-lexi-order]
The \textit{PR-lexi-order} is the total order $(2^N,\preceq_{\mathsf{pr}})$, where for two coalitions $S \subseteq N$ and $S' \subseteq N$:
$S \preceq_{\mathsf{pr}} S'$ if and only if $\mathsf{pr}(S)$ \textit{lexicographically precedes} $\mathsf{pr}(S')$.
A vector $\vec{v}$ lexicographically precedes another vector $\vec{v'}$ when there exists a $i$ such that $v_i < v_i'$ and for all $j < i$ it holds that $v_i = v_j$.
\end{definition}
For example, we see that for $N = \{1,\ldots,5\}$, we have $\{1,2,3\} \preceq_{\mathsf{pr}} \{1,3,5\}$. The least element of $(2^N,\preceq_{\mathsf{pr}})$ is $\varnothing$ and the greatest element of $(2^N,\preceq_{\mathsf{pr}})$ is $N$.

Next, we introduce some operations that we can apply to coalitions.
For this, the reader should recall definitions \ref{leftshiftrightshift} and \ref{characteristicvector}.
\begin{definition}[fill-up, bottom right-shift, truncation, immediate successor]\label{coalitionoperations}
Let $N$ be the set of players $\{1,\ldots,n\}$ and let $S \subseteq N$ be a coalition. The functions $a$ and $b$ are defined as follows.
\begin{itemize}
\item $b(S)$ is the largest index $j$ such that $\chi(j,S) = 1$.
\item $a(S)$ is the largest index $j$ such that $\chi(j,S) = 0$ and $\chi(j+1,S) = 1$ (if such a $j$ does not exist, then $a(S) = 0$).
\end{itemize}

Now we can define the following operations on $S$:
\begin{itemize}
\item The \textit{fill-up} of $S$: $\mathsf{fill}(S) = S \cup \{b(S)+1\}$ (undefined if $S = N$).
\item The \textit{bottom right-shift} of $S$: $\mathsf{brs}(S) = S \cup \{b(S)+1\} \setminus \{b(S)\}$ (undefined if $b(S) = n$).
\item The \textit{truncation} of $S$: $\mathsf{trunc}(S) = S \setminus \{a(S)+1, \ldots, n\}$.
\item The \textit{immediate successor} of $S$:
\begin{equation*}
\mathsf{succ}(S) = \begin{cases} \mathsf{fill}(S) \text{ if } n \not\in S \text{ , } \\ \mathsf{brs}(S \setminus \{n\}) \text{ if } n \in S \text{ . } \end{cases}
\end{equation*}
\end{itemize}
\end{definition}
The immediate successor operation is named as such because it denotes the successor of $S$ in the total order $(2^N,\preceq_{\mathsf{pr}})$.

One last concept we need is that of a \textit{shelter} coalition.
\begin{definition}[Shelter]
A \textit{shelter} is a minimal winning coalition $S$ such that $\mathsf{brs}(S)$ is losing or undefined.
\end{definition}
Note that the set of roof coalitions of a canonical linear game is a subset of the set of shelter coalitions of that game.

\paragraph{The Hop-Skip-and-Jump algorithm}\label{actualhopskipjumpalgorithm}
We are now ready to state the algorithm.
The input to the algorithm is a string $\ell$ in $\mathcal{L}_{W,\min}$, i.e., the list of characteristic vectors describing the set of minimal winning coalitions $W_{\min}$.
The four main steps of the algorithm are:
\begin{enumerate}
\item Check whether $G_\ell$ is a linear game. If not, then stop. When it turns out that the game is linear, find a permutation of the players that turns the game into a canonical linear game. In the remaining steps, we assume that $G_{\ell}$ is a canonical linear game.
\item Generate a list of shelters $\mathcal{S}$, sorted according to the PR-lexi-order.
\item Use $\mathcal{S}$ as input for the \textit{Hop-Skip-and-Jump} algorithm. The Hop-Skip-and-Jump algorithm will give as output the set of all maximal losing coalitions $L_{\max}$. This step, is the most non-trivial part, and we will explain it in detail below.
\item Use $W_{\min}$ and $L_{\max}$ to generate the following system of linear inequalities, and solve it for any choice of $q$ in order to find the weights $w_1,\ldots,w_n$:
\begin{equation}\label{systemofinequalities}
\begin{split}
w_1 \chi(1,S) + \cdots + w_n \chi(n,S) \geq q , \forall S \in W_{\min} \\
w_1 \chi(1,S) + \cdots + w_n \chi(n,S) < q , \forall S \in L_{\max}
\end{split}
\end{equation}
If this system of linear inequalities has no solutions, then $G_\ell$ is not weighted; and otherwise the weights that have been found are the weights of the players, and $q$ is the quota: $[q; w_1, \ldots, w_n]$ is a weighted form of the weighted voting game.
\end{enumerate}

The first step of the algorithm is easy if we use an algorithm by Aziz, given in \cite{harisazizinfluence}. This algorithm decides whether a monotonic simple game represented as a listing of minimal winning coalitions is a linear game, and if so it outputs a strict desirability order\footnote{With this, we mean that the algorithm outputs a list $\vec{P} = (P_1, \ldots, P_j)$ such that $\{P_1, \ldots, P_j\}$ is a partition of $N$, where the players of a set in this partition are all equally desirable, and for all $i$ and $j$ with $i > j$ we have that any player in $P_i$ is strictly more desirable than any player in $P_j$.}. From the strict desirability order, the required permutation directly follows.

The generation of the sorted list of shelters can be done in polynomial-time: We can easily check for each minimal winning coalition whether its bottom right-shift is losing.

Linear programs are solvable in a time that is polynomial in the size of the linear program, by Karmarkar's algorithm \cite{karmarkar} for example.
For the linear program of the fourth part of the algorithm we will have to show that its size is bounded by a polynomial in $n$ and the number of minimal winning coalitions, i.e., we will have to show that there are only polynomially many more maximal losing coalitions than that there are minimal winning coalitions. This follows from the fact that the Hop-Skip-and-Jump algorithm (see below) runs in polynomial time and hence can output only a polynomial number of coalitions. Lastly, the fact that we can choose any $q \in \mathbb{R}_{> 0}$ follows from Theorem \ref{wvginvmult}.

The hard part that now remains is part three of the algorithm: outputting the list of maximal losing coalitions, given a sorted list of shelter coalitions. This is what the Hop-Skip-and-Jump-algorithm does. We will now state this algorithm, prove it correct and show that the runtime is bounded by a polynomial in the number of players $n$ and the number of shelter coalitions $t$. From this polynomial runtime it then also follows that $|L_{\max}|$ is polynomially bounded in $|W_{\min}|$.

The pseudocode for the Hop-Skip-and-Jump algorithm is given in Algorithm \ref{alg:hopskipjump}.
The basic idea is to consider all coalitions in the order induced by the PR-lexi-order, and output those coalitions that are maximal losing coalitions. During this process, we will be able to skip huge intervals of coalitions in order to achieve a polynomial run-time.

\begin{algorithm}
\caption[The Hop-Skip-and-Jump algorithm.]{The Hop-Skip-and-Jump algorithm. A polynomial-time algorithm that outputs the set of maximal losing coalitions of a monotonic simple game $G$ on players $N = \{1,\ldots,n\}$, given the sorted list of shelters of $G$ as input. An assumption we make in this algorithm is that the empty coalition does not occur in the list of shelters. If it does, it becomes a trivial task to output the list of maximal losing coalitions, so this is a safe assumption.}
\label{alg:hopskipjump}
\begin{algorithmic}[1]
\STATE $\mathsf{nextshelter} := $ first shelter on the list. \COMMENT{Output the coalition $N$ and stop if the list is empty.}
\STATE $\mathsf{currentcoalition} := \varnothing$ \COMMENT{Start with the least coalition, according to the PR-lexi-order.}
\LOOP
\WHILE{$\mathsf{currentcoalition} \not= \mathsf{nextshelter} \setminus \{b(\mathsf{nextshelter})\}$}
\IF{$n \not\in \mathsf{currentcoalition}$}
\STATE $\mathsf{currentcoalition} := \mathsf{fill}(\mathsf{currentcoalition})$
\ELSE
\STATE \textbf{output} $\mathsf{currentcoalition}$
\STATE $\mathsf{currentcoalition} := \mathsf{brs}(\mathsf{trunc}(\mathsf{currentcoalition}))$ \COMMENT{Stop if undefined.}
\ENDIF
\ENDWHILE
\IF{$n \not\in \mathsf{nextshelter}$}
\STATE $\mathsf{currentcoalition} := \mathsf{brs}(\mathsf{nextshelter})$
\ELSE
\STATE \textbf{output} $\mathsf{currentcoalition}$
\STATE $\mathsf{currentcoalition} := \mathsf{succ}(\mathsf{nextshelter})$ \COMMENT{Stop if $\mathsf{nextshelter} = \{n\}$.}
\ENDIF
\STATE $\mathsf{nextshelter} := $ next shelter on the list.
\ENDLOOP
\end{algorithmic}
\end{algorithm}

We will now proceed by giving a correctness-proof of this algorithm.
\begin{theorem}\label{hopskipjumponlymlcs}
Algorithm \ref{alg:hopskipjump} outputs only maximal losing coalitions.
\end{theorem}
\begin{proof}
There are three places at which Algorithm \ref{alg:hopskipjump} outputs coalitions: line 1, 8 and line 15.

At line 1, a coalition is only output when the list of shelters is empty. When this list is empty, it means there are no winning coalitions, so $N$ is the only maximal losing coalition.

At line 15 we see that $\mathsf{currentcoalition} \subset \mathsf{nextshelter}$, and $\mathsf{nextshelter}$ is a minimal winning coalition, so $\mathsf{currentcoalition}$ must be losing. Also, at line 15, $n \in \mathsf{nextshelter}$. This means that any superset of $\mathsf{currentcoalition}$ is a superset of a leftshift of $\mathsf{nextshelter}$, and therefore winning. So we conclude that $\mathsf{currentcoalition}$ is a maximal losing coalition. This establishes that at line 15, all coalitions output are maximal losing coalitions.

Now we need to show the same for line 8. For this, we first need to prove the following invariant.
\begin{lemma}\label{hopskipjumplemma}
When running Algorithm \ref{alg:hopskipjump}, directly after executing line 2, line 18, and each iteration of the while-loop of line 4, $\mathsf{currentcoalition}$ is a losing coalition.
\end{lemma}
\begin{proof}
We prove all three cases separately.
\begin{itemize}
\item \textit{Directly after executing line 2}, $\mathsf{currentcoalition}$ is the empty coalition and thus losing by assumption.
\item \textit{Directly after executing line 18}, we have two subcases:
\begin{description}
\item[Case 1:] \textit{After the last time the execution of the algorithm passed line 11, lines 12 and 13 were executed while lines 14--16 were skipped.} In this case, $\mathsf{currentcoalition}$ is a bottom right-shift of a shelter, so $\mathsf{currentcoalition}$ is losing by the definition of a shelter.
\item[Case 2:] \textit{After the last time the execution of the algorithm passed line 11, lines 14--16 were executed while lines 12 and 13 were skipped.} In this case, $\mathsf{currentcoalition}$ is a direct successor of a shelter $s$ containing player $n$, by definition of the direct successor function, $\mathsf{currentcoalition}$ is a subset of $s$ and hence losing.
\end{description}
\item \textit{Directly after each iteration of the while-loop of line 4.} We can use induction for this final case. By the preceding two cases in this list, that we proved, we can assume that $\mathsf{currentcoalition}$ is losing when the while-loop is entered. It suffices now to show that $\mathsf{currentcoalition}$ is losing after a single repetition of the while-loop. We divide the proof up again, in two cases:
\begin{description}
\item[Case 1:] \textit{During the execution of the while-loop, lines 5 and 6 were executed while lines 7--9 were skipped.} Then $\mathsf{currentcoalition}$ is a fill-up of a losing coalition, say $l$. Let $i$ be the agent that was added by the fill-up, i.e., $\mathsf{currentcoalition} = l \cup \{i\}$. Suppose for contradiction that $\mathsf{currentcoalition}$ is winning; then $i \in \mathsf{nextshelter}$ and $i-1 \in \mathsf{nextshelter}$. It must also be true that $l \subseteq \mathsf{nextshelter}$ because otherwise $l$ is a right-shift of $\mathsf{nextshelter}$ and therefore winning (the induction hypothesis states that $l$ is losing). Therefore $l = \mathsf{nextshelter} \setminus \{b(\mathsf{nextshelter})\}$. But then execution would have left the loop because of line 4. Contradiction.
\item[Case 2:] \textit{In the execution of the while-loop, lines 7--9 were executed while lines 5 and 6 were skipped.}
$\mathsf{currentcoalition}$ is a bottom right-shift of a truncation of a losing coalition. A truncation of a losing coalition is losing, and a bottom right-shift of a losing coalition is losing, so $\mathsf{currentcoalition}$ is losing.
\end{description}
\end{itemize}
\end{proof}
From the lemma above, it follows that at line 8, $\mathsf{currentcoalition}$ is losing. To show that it is also maximal, we divide the proof up in three cases:
\begin{description}
\item[Case 1:] \textit{The execution of the algorithm has never passed line 11.} In this case $\mathsf{currentcoalition}$ at line 8 is obtained by a series of successive fill-ups starting from the empty coalition, and $n \in \mathsf{currentcoalition}$. This means that $\mathsf{currentcoalition} = N$, so $\mathsf{currentcoalition}$ is maximal.
\item[Case 2:] \textit{The execution of the algorithm did pass line 11 at least once, and the last time that execution has done so lines 12 and 13 were executed while lines 14--16 were skipped.} In this case we have that at line 8, $\mathsf{currentcoalition}$ is obtained by a series of fill-ups of a bottom right-shift of a shelter-coalition $s$. It follows that adding any player to $\mathsf{currentcoalition}$ will turn $\mathsf{currentcoalition}$ into a winning coalition, because $\mathsf{currentcoalition}$ would then become a superset of a left-shift of $s$. So $\mathsf{currentcoalition}$ is maximal.
\item[Case 3:] \textit{The execution of the algorithm did pass line 11 at least once, and the last time that execution has done so, lines 14--16 were executed while lines 12 and 13 were skipped.} In this case we have at line 8 that $\mathsf{currentcoalition}$ is the successor of a shelter $s$ that has player $n$ in it. By the definition of the successor function we get that adding any player to $\mathsf{currentcoalition}$ would make it a superset of a left-shift of $s$, and thus winning. So $\mathsf{currentcoalition}$ is maximal.
\end{description}
\end{proof}

\begin{theorem}
Algorithm \ref{alg:hopskipjump} outputs all maximal losing coalitions.
\end{theorem}
\begin{proof}
By Theorem \ref{hopskipjumponlymlcs} we have that Algorithm \ref{alg:hopskipjump} outputs only maximal losing coalitions, so what suffices is to show that the intervals of coalitions that Algorithm \ref{alg:hopskipjump} does not output, do not contain any losing coalitions.

Let $s$ be a coalition that is not output by Algorithm \ref{alg:hopskipjump}. There are several cases possible.
\begin{description}
\item[Case 1:] \textit{There is a point when the execution of the algorithm has just passed line 6, such that $\mathsf{currentcoalition} = s$.} In that case $s$ is losing, following from Lemma \ref{hopskipjumplemma}.
\item[Case 2:] \textit{There is a point when the execution of the algorithm has just passed line 8, such that $\mathsf{currentcoalition} \preceq_{\mathsf{pr}} s \preceq_{\mathsf{pr}} \mathsf{brs}(\mathsf{trunc}(\mathsf{currentcoalition}))$.} Now $s$ is a direct right-shift of a point $s'$ that the algorithm has output. $s'$ is maximal losing so $s$ is not maximal losing.
\item[Case 3:] \textit{There is a point when the execution of the algorithm has just passed line 12, such that $\mathsf{currentcoalition} \preceq_{\mathsf{pr}} s \preceq_{\mathsf{pr}} \mathsf{brs}(\mathsf{nextshelter})$.} Here we have that $s$ is either a right-shift of $\mathsf{currentcoalition}$ or a left-shift of a superset of $\mathsf{nextshelter}$. In the former case, $s$ is not a maximal losing coalition because it is a right-shift of $\mathsf{currentcoalition}$, and $\mathsf{currentcoalition}$ is not a maximal losing coalition because it is a strict subset of the bottom right-shift of $\mathsf{nextshelter}$, which is also losing. In the latter case, $s$ is winning, so $s$ can not be maximal losing.
\end{description}
\end{proof}

By the two theorems above, we have established that the Hop-Skip-and-Jump algorithm works correctly. Now we will also show that it runs in polynomial time.

\begin{theorem}
Algorithm \ref{alg:hopskipjump} runs in time $O(n^3t)$ (where $t$ is the number of shelter coalitions).
\end{theorem}
\begin{proof}
When repeatedly executing the while-loop of line 4, lines 5 and 6 can be executed only $n$ consecutive times, before lines 7--9 are executed. Line 9 can be executed at most $n$ times in total, given that the execution does not leave the while-loop (after $n$ times, the operation done at line 9 is undefined, and execution stops). It follows that the while-loop is executed at most $n^2$ consecutive times before execution leaves the while-loop. Each time lines 12--18 are executed, one shelter is taken from the list, so lines 12--18 are executed only $t$ times. The fill-up operation, bottom right-shift operation, successor operation and truncation operation can all be implemented in $O(n)$ time. So, bringing everything together, we arrive at a total runtime of $O(n^3t)$.
\end{proof}

\subsubsection{Synthesizing roof- and ceiling-representations}\label{roofvgs}
Next, we consider the problem of synthesizing various representations of games into the roof- and ceiling-representation of a canonical linear game.

Let us start with the problem ($\mathcal{L}_{\mathsf{W}}, \mathcal{L}_{\mathsf{roof}}$)-VGS. This problem boils down to solving the ($\mathcal{L}_{\mathsf{W,\min}}, \mathcal{L}_{\mathsf{roof}}$)-VGS problem, since ($\mathcal{L}_{\mathsf{W}}, \mathcal{L}_{\mathsf{W,\min}}$)-VGS is easy (just check for each coalition in $W$ whether it is minimal, and if so, it is in $W_{\min}$). The same holds for the problems ($\mathcal{L}_{\mathsf{L}}, \mathcal{L}_{\mathsf{ceil}}$)-VGS and ($\mathcal{L}_{\mathsf{L,\max}}, \mathcal{L}_{\mathsf{ceil}}$)-VGS.

Solving ($\mathcal{L}_{\mathsf{W,\min}}, \mathcal{L}_{\mathsf{roof}}$)-VGS is also not very difficult. As pointed out before, there is a polynomial-time algorithm that checks whether a monotonic simple game given as a list of minimal winning coalitions is linear, and we can obtain the strict desirability order if this is the case. It could be that it turns out the game is linear, but not canonical. If we wish, we are then also able to permute the players so that we end up with a canonical linear game.
After that, all that we have to do is check for each minimal winning coalition $C$ whether each of its direct right-shifts (no more than $n$ direct right-shifts are possible) are losing coalitions. If that is the case, then $C$ must be a roof. For the problem ($\mathcal{L}_{\mathsf{L,\max}}, \mathcal{L}_{\mathsf{ceil}}$)-VGS, the situation is completely symmetric.

What also follows now, is that the problems ($\mathcal{L}_{\mathsf{W,\min}}, \mathcal{L}_{\mathsf{ceil}}$)-VGS can be solved in polynomial time: we first check if the input list of minimal winning coalitions describes a linear game. If so, then the Hop-Skip-and-Jump algorithm of Section \ref{actualhopskipjumpalgorithm} is able to generate in polynomial time a list of maximal losing coalitions from the list of minimal winning coalitions. After that, we filter from this output list the coalitions that are not ceiling coalitions. The problem ($\mathcal{L}_{\mathsf{L,\max}}, \mathcal{L}_{\mathsf{roof}}$)-VGS is also solvable in polynomial time by running a ``symmetric'' version of the Hop-Skip-and-Jump algorithm where we
\begin{itemize}
\item permute the players according to the permutation $\pi$ where the players are ordered in \emph{ascending} desirability, i.e., the least desirable player is now player 1, and the most desirable player is player $n$;
\item run a version of the Hop-Skip-and-Jump algorithm where losing coalitions are treated as winning coalitions and vice versa.
\end{itemize}
once the Hop-Skip-and-Jump algorithm is done, we have a list of coalitions. For each $C$ in this list, the coalition $\{\pi^{-1}(i) : i \in C\}$ is a minimal winning coalition.

As a consequence, by polynomial time solvability of ($\mathcal{L}_{\mathsf{W}}, \mathcal{L}_{\mathsf{W,\min}}$)-VGS and ($\mathcal{L}_{\mathsf{L}}, \mathcal{L}_{\mathsf{L,\max}}$)-VGS we also have that ($\mathcal{L}_{\mathsf{L}}, \mathcal{L}_{\mathsf{roof}}$)-VGS and ($\mathcal{L}_{\mathsf{W}}, \mathcal{L}_{\mathsf{roof}}$)-VGS admit a polynomial time algorithm.

Is the problem ($\mathcal{L}_{\mathsf{ceil}}, \mathcal{L}_{\mathsf{roof}}$)-VGS solvable in polynomial time? This turns out to not be the case. We will now give a family of examples of canonical linear games in which the number of roof coalitions is exponential in $n$, while the number of ceiling coalitions is only polynomial in $n$. As a consequence, any algorithm that generates the list of roofs from the list of ceilings will run in exponential time in the worst case. By symmetry it also follows that ($\mathcal{L}_{\mathsf{roof}}, \mathcal{L}_{\mathsf{ceil}}$)-VGS is not solvable in polynomial time.

Let us first define the following specific type of coalition.

\begin{definition}[$(k,i)$-encoding coalition]
Let $N = \{1,\ldots, n\}$ be a set of players such that $n = 4i$ for some $i \in \mathbb{N}$.
For any $k$ satisfying $0 \leq k < 2^i-1$, the \textit{$(k,i)$-encoding coalition} $S_{k,i} \subseteq N$ is then defined as
\begin{eqnarray*}
\{4(j-1)+2, 4(j-1)+3 : \text{The } j\text{th bit in the binary representation of } k \text{ equals } 0 \text{.}\} & \cup & \\
\{4(j-1)+1, 4(j-1)+4 : \text{The } j\text{th bit in the binary representation of } k \text{ equals } 1 \text{.}\}
\end{eqnarray*}
\end{definition}
For example, $S_{2,2} = \{1,4,6,7\}$, and $S_{5,3} = \{1,4,6,7,9,12\}$.
We can then define canonical linear games in which the roof coalitions are $(k,i)$-encoding coalitions.
\begin{definition}[$i$-bit roof game]
Let $N = \{1,\ldots, n\}$ be a set of players such that $n = 4i$ for some $i \in \mathbb{N}$.
The \textit{$i$-bit roof game} on $N$, denoted $G_{i\mathsf{-bit}}$, is the canonical linear game such that the set of roof coalitions of $G$ is
$\{S_{0,i}, \ldots, S_{2^i-1,i}\}$.
\end{definition}
For example, the $2$-bit roof game, $G_{2\mathsf{-bit}}$, consists of the roofs $\{\{2,3,6,7\}$, $\{2,3,5,8\}$, $\{1,4,6,7\}$, $\{1,4,5,8\}\}$.
$G_{i\mathsf{-bit}}$ is well-defined for all $i$ because the binary representations of two arbitrary $i$-bit numbers $k$ and $k'$ differ in at least one bit. Therefore, $S_{i,k}$ is not a superset of a left-shift of $S_{i,k'}$ and hence the set of roofs that we have defined for $G_{i\mathsf{-bit}}$ is indeed a valid set of roofs (i.e., there are no two roofs such that one is a left-shift of another).

$G_{i\mathsf{-bit}}$ has $2^i = 2^{\frac{n}{4}}$ roofs, i.e., an exponential number in $n$. We will show that the number of ceilings in $G_{i\mathsf{-bit}}$ is only polynomially bounded. First let us use the following definitions for convenience.
\begin{definition}[Accepting roof set]
Let $G \in \mathcal{G}_{\mathsf{clin}}(n)$ be a canonical linear game on players $N = \{1,\ldots, n\}$. Let $C \subseteq N$ be a coalition, let $x$ be a natural number such that $1 \leq x \leq |C|$, and let $D(C,x)$ be the $x$-th most desirable player in $C$.
The \textit{accepting set of roofs of the $x$-th most desirable player in $C$}, denoted $A(C, x)$, is the set consisting of those roof coalitions $R$ for which either the $x$th most desirable player in $R$ is greater than or equal to $D(c,x)$, or $|R| < x$.
\end{definition}

It is important to now observe that the following fact holds.
\begin{proposition}
In a canonical linear game, a coalition $C$ is winning if and only if $\bigcap_{a=1}^{|C|} A(C,a) \not= \varnothing$.
\end{proposition}
\begin{proof}
This lemma is in fact an equivalent statement of the fact that $C$ is winning in a canonical linear game if and only if it is a superset of a left-shift of a roof: if $R \in \bigcap_{a=1}^{|C|} A(C,a)$ then it means that replacing the $a$-th most desirable player in $R$ by the $a$-th most desirable player in $C$ for all $a$,$1 \leq a \leq R$ would result in a left-shift of $R$ that is a subset of $C$, so $C$ must be winning.

Conversely, suppose $C$ is winning. Then there must be a roof $R$ that is a right-shift of a subset of $C$. By removing from $C$ the players with a higher number than $D(C,|R|)$, we obtain a subset $C'$ of $C$ with $|R|$ players. By replacing the $a$-th most desirable player of $C$ by the $a$-th most desirable player of $R$ for $1 \leq a \leq R$, we obtain a right-shift of $C$ that is $R$. Because in this last step we replaced each player in $C'$ by a higher-numbered player, we get that $R \in \bigcap_{a=1}^{|R|} A(C,a)$. $R$ is also in $\bigcap_{a = |R|+1}^{|C|} A(C,a)$ by definition.
\end{proof}

Using the notion of an accepting roof set, we can prove the following technical lemma. The reader should recall the definition of a \emph{direct} left-shift (Definition \ref{leftshiftrightshift}).
\begin{lemma}\label{technicallemma}
Let $C$ be a ceiling of $G_{i\mathsf{-bit}}$ with two or more distinct coalitions that are direct left-shifts of $C$, and let $p$ be an arbitrary player that we can apply the direct left-shift operation on, i.e., let $p$ be a player such that $C_1 = C \cup \{p-1\} \setminus \{p\}$ is a direct left-shift of $C$. Also, let $a$ be the number such that $p = D(C,a)$. Then $p = 2a$.
\end{lemma}
\begin{proof}
Observe that for all $b$ it holds that every roof $R$ of $G_{i\mathsf{-bit}}$ has either $D(R,b) = 2b-1$ or $D(R,b) = 2b$.
By construction of $G_{i\mathsf{-bit}}$, the number of roofs of $G_{i\mathsf{-bit}}$ that contain player $2b-1$ is $\frac{2^i}{2}$, and the number of roofs that contain player $2b$ is also $\frac{2^i}{2}$.

$C$ has at least two distinct direct left-shifts, so there must be another player $p'$, $p' \not= p$, such that $C_2 = C \cup \{p'-1\} \setminus \{p'\}$ is a direct left-shift of $C$.

First we will show that $p \leq 2a$. Assume therefore that $p > 2a$. Now we have that $|A(C,a)| = 0$, so then $|A(C_2,a)| = 0$ and hence $\bigcap_a A(C_2,a) = \varnothing$. We see that $C_2$ is losing, but $C_2$ is a direct left-shift of $C$, which is a ceiling, so $C_2$ is winning. This is a contradiction, so $p \leq 2a$.

Now we will show that $p \geq 2a$. Assume therefore that $p < 2a$. Now we have that $|A(C,a)| = 2^i$, so then $A(C_1,a) = 2^i$. Now it must be that $\bigcap_a A(C_1,a) = \bigcap_a A(C,a)$. But $\bigcap_a A(C,a) = \varnothing$ because $C$ is losing, and therefore $\bigcap_a A(C_1,a) = \varnothing$ so $C_1$ is losing. $C_1$ is also winning, because it is a left-shift of ceiling $C$. This is a contradiction, so $p \geq 2a$.

$p \geq 2a$ and $p \leq 2a$, so $p = 2a$.
\end{proof}

\begin{lemma}
In $G_{i\mathsf{-bit}}$, a ceiling does not have more than two direct left-shifts.
\end{lemma}
\begin{proof}
For contradiction, let $C$ be a ceiling with more than two direct left-shifts. Let $k$ be the number of direct left-shifts of $C$, and let
$P = \{p_1, \ldots, p_k\}$ be the set containing the players of $C$ that we can apply the direct left-shift operation on (we say that we can apply the direct left-shift operation on a player $q$ when $C \cup \{q-1\} \setminus \{q\}$ is a left-shift of $C$).
Let $\mathcal{A} = \{a_1, \ldots, a_k\}$ then be the numbers such that $p_j$ is the $a_j$-th most desirable player in $C$, for all $i$ with $1 \leq j \leq k$.
For any $j \in \{1,\ldots,i\}$ and any $b \in \{0,1\}$, let $R(j,b)$ denote the following set of roofs of $G_{i\mathsf{-bit}}$:
\begin{equation*}
R(j,b) = \{S_{k,i} : \text{The } j\text{-th bit of the binary representation of } k \text{ is } b\text{. }\}
\end{equation*}
Observe that by the previous lemma, there is a $k$-tuple of bits $(b_1,\ldots,b_k) \in \{0,1\}^k$ such that for all $j$ with $1 \leq j \leq k$:
\begin{equation*}
A(C,a_j) = R(\lceil p_j / 4 \rceil, k_j) .
\end{equation*}
There are now two cases:
\begin{description}
\item[Case 1:] \textit{All of the players $\{p_1, \ldots, p_k\}$ are in different multiples of 4, i.e., $\lceil p_1 / 4 \rceil \not= \lceil p_2 / 4 \rceil \not= \cdots \not= \lceil p_k / 4 \rceil$.} Then by the properties of the binary numbers, the intersection $\bigcap_{a \in \mathcal{A}} A(C,a) = \bigcap_{p \in P} R(\lceil p / 4 \rceil, b)$ is not empty, therefore $C$ must be winning, which is in contradiction with $C$ being a ceiling. So this case is impossible.
\item[Case 2:] \textit{There are two distinct players $p$ and $p'$, both in $P$, that are in the same multiple of 4, i.e., $\lceil p / 4 \rceil = \lceil p' / 4 \rceil$.}
Assume without loss of generality that $p < p'$. Then $A(C,a) \cap A(C,a') = \varnothing$. But then we would be able to apply a direct left-shift on player $p''$ without turning $C$ into a winning coalition, i.e., $C \cup \{p''-1\} \setminus \{p''\}$ is winning. But $C$ is a ceiling, so that is a contradiction.
\end{description}
From the previous lemma it follows that there can not be more than two players that are the same multiple of 4, so the above two cases are indeed exhaustive.
Both cases are impossible, so we must reject the assumption that there exists a ceiling $C$ with more than two left-shifts.
\end{proof}

It is easy to see that there exist no more than $O(n^5)$ coalitions with exactly two left-shifts, there are no more than $O(n^3)$ coalitions with one left-shift, and there are no more than $O(n)$ coalitions with no left-shifts. so we get the following corollary.
\begin{corollary}
The game $G_{i\mathsf{-bit}}$ (on $n=4i$ players) has $O(n^5)$ ceilings.
\end{corollary}
We can now conclude that $\{G_{i\mathsf{-bit}} : i \in \mathbb{N} \}$ is an infinite family of examples in which there are exponentially many more roofs than ceilings. Hence, finally we obtain:
\begin{corollary}
there is no polynomial time algorithm for ($\mathcal{L}_{\mathsf{ceil}}, \mathcal{L}_{\mathsf{roof}}$)-VGS \\
and ($\mathcal{L}_{\mathsf{ceil}}, \mathcal{L}_{\mathsf{roof}}$)-VGS.
\end{corollary}

\subsubsection{Other voting game synthesis problems \& summary of complexity results for voting game synthesis}\label{othervgs}
In this section we will discuss some of the remaining variants of the voting game synthesis problem that we did not discuss in the other sections.
At the end of this section, Table \ref{ch4summary} summarizes all of the results that we have discussed up to now.

First of all, Freixas et al.\ investigate in \cite{freixasproperties} the $(\mathcal{L}_1, \mathcal{L}_2)$-VGS problem for $\mathcal{L}_1$ and $\mathcal{L}_2 \in \{\mathcal{L}_W, \mathcal{L}_{W,\min}, \mathcal{L}_L, \mathcal{L}_{L,\max}\}$.
Most of their results follow from the discussion above. One of their results that does not, is that $(\mathcal{L}_W, \mathcal{L}_L)$-VGS and $(\mathcal{L}_L, \mathcal{L}_W)$-VGS do not have a polynomial time algorithm. This holds because there are instances where there are exponentially many more losing coalitions than that there are winning coalitions. Consider for instance the game in which only the grand coalition is winning. In this game there are $2^n-1$ losing coalitions, so it takes exponential time to list them all. This game is also a canonical linear game and a weighted voting game, so even if we restrict games to be weighted, or canonical linear, it still holds that $(\mathcal{L}_W, \mathcal{L}_L)$-VGS and $(\mathcal{L}_L, \mathcal{L}_W)$-VGS do not have polynomial time algorithms.

In \cite{freixasproperties}, it is also shown that $(\mathcal{L}_{W, \max}, \mathcal{L}_{L,\min})$-VGS is in general not polynomial time solvable. The authors show this by giving a family of examples of monotonic simple games that have exponentially many more maximal losing coalitions than minimal winning coalitions. The Hop-Skip-and-Jump algorithm that we described above does actually solve $(\mathcal{L}_{W, \max}, \mathcal{L}_{L,\min})$-VGS in polynomial time, but only for the restriction to linear games.

Another set of voting game synthesis problems that we have not yet discussed is the $(\mathcal{L}_{\mathsf{weights}}, \mathcal{L}_2)$-VGS case, for any choice of $\mathcal{L}_2$. In this case it always holds that there is no polynomial time algorithm for the problem:
\begin{itemize}
\item When $\mathcal{L}_2 = \mathcal{L}_W$, consider the weighted voting game in which the quota is $0$. Now there are $2^n$ minimal winning coalitions, so the output is exponentially larger than the input. The case $\mathcal{L}_2 = \mathcal{L}_L$ is analogous, but now we take a weighted voting game in which the quota is larger than the sum of all weights, so that there are no winning coalitions.
\item When $\mathcal{L}_2 = \mathcal{L}_{W,\min}$ or $\mathcal{L}_2 = \mathcal{L}_{L,\max}$, we see that the weighted voting game in which every player's weight is 1 and the quota is $\lfloor \frac{n}{2} \rfloor$ has an exponential number of minimal winning coalitions and maximal losing coalitions: any coalition of size $\lfloor \frac{n}{2} \rfloor$ is minimal winning, and any coalition of size $\lfloor \frac{n}{2} \rfloor - 1$ is maximal losing. There are respectively $\binom{n}{\lfloor n/2 \rfloor}$ and $\binom{n}{\lfloor n/2 \rfloor - 1}$ such coalitions (by Sperner's theorem, see Theorem \ref{spernerstheorem}). By using Stirling's approximation, we can see that both these expressions are exponential in $n$.
\item When $\mathcal{L}_2 = \mathcal{L}_{\mathsf{roof}}$ it follows directly from the proof of Theorem \ref{lineargameslowerbound} that the weighted voting game in which player $i$ gets weight $n-i+1$ and the quota is equal to $n(n+1)/4$, has an exponential number of roofs. We do not know whether there is also a weighted voting game with an exponential number of ceilings.
\end{itemize}

That completes our study of the voting game synthesis problem. As said before, table \ref{ch4summary} summarizes all of the results that we have discussed and obtained. It indicates for each variant of the voting game synthesis problem whether it is solvable in polynomial time ($\mathsf{P}$), or does not have a polynomial time algorithm ($\mathsf{EXP}$). We see that the complexities of three problems remain open:
\begin{itemize}
\item transforming roof-representations of canonical linear games into weighted representations,
\item transforming ceiling-representations of canonical linear games into weighted representations,
\item transforming weighted-representations of weighted voting games into ceiling representations.
\end{itemize}

\begin{table}[htbp]
\caption{Time complexities of the various $(\mathcal{L}_1, \mathcal{L}_2)$-VGS problems that we have discussed in this section.}
\label{ch4summary}
\begin{center}
\begin{tabular}{r||c|c|c|c|c|c|c|}
$\mathcal{L}_2 \rightarrow$ & $\mathcal{L}_W$ & $\mathcal{L}_{W,\min}$ & $\mathcal{L}_L$ & $\mathcal{L}_{L,\max}$ & $\mathcal{L}_{\mathsf{roof}}$ & $\mathcal{L}_{\mathsf{ceil}}$ & $\mathcal{L}_{\mathsf{weights}}$ \\
$\mathcal{L}_1 \downarrow$ & & & & & & & \\
\hline\hline
$\mathcal{L}_W$ & - & $\mathsf{P}$ & $\mathsf{EXP}$ & $\mathsf{P}$ & $\mathsf{P}$ & $\mathsf{P}$ & $\mathsf{P}$ \\
\hline
$\mathcal{L}_{W,\min}$ & $\mathsf{EXP}$ & - & $\mathsf{EXP}$ & $\mathsf{EXP}$ ($\mathsf{P}$ if linear) & $\mathsf{P}$ & $\mathsf{P}$ & $\mathsf{P}$ \\
\hline
$\mathcal{L}_L$ & $\mathsf{EXP}$ & $\mathsf{P}$ & - & $\mathsf{P}$ & $\mathsf{P}$ & $\mathsf{P}$ & $\mathsf{P}$ \\
\hline
$\mathcal{L}_{L,\max}$ & $\mathsf{EXP}$ & $\mathsf{EXP}$ ($\mathsf{P}$ if linear) & $\mathsf{EXP}$ & - & $\mathsf{P}$ & $\mathsf{P}$ & $\mathsf{P}$ \\
\hline
$\mathcal{L}_{\mathsf{roof}}$ & $\mathsf{EXP}$ & $\mathsf{EXP}$ & $\mathsf{EXP}$ & $\mathsf{EXP}$ & - & $\mathsf{EXP}$ & $\mathsf{?}$ \\
\hline
$\mathcal{L}_{\mathsf{ceil}}$ & $\mathsf{EXP}$ & $\mathsf{EXP}$ & $\mathsf{EXP}$ & $\mathsf{EXP}$ & $\mathsf{EXP}$ & - & $\mathsf{?}$ \\
\hline
$\mathcal{L}_{\mathsf{weights}}$ & $\mathsf{EXP}$ & $\mathsf{EXP}$ & $\mathsf{EXP}$ & $\mathsf{EXP}$ & $\mathsf{EXP} $ & $\mathsf{?}$ & - \\
\hline
\end{tabular}
\end{center}
\end{table}

\section{Solving the power index voting game design problem}\label{sec:PVGDsolving}

From the existing literature on the power index voting game design problem, we see that researchers have only considered heuristic methods for the case where a weighted representation must be output. Even stronger: the weighted representation is the only representation that current voting game design algorithms \textit{internally} work with. No other methods of representing a game have even been considered.

There exists an \emph{infinite} number of weighted representations for each weighted voting game (this follows from Proposition~\ref{wvginvmult}). This makes it hard to derive an exact algorithm that is based on working with weighted representations alone, since there is no clear finite set of weight vectors that an algorithm can search through.\footnote{However, the literature does provide us with bounds on the maximum weight necessary in an integer representation of a weighted voting game, and we could utilize this in order to come up with an enumeration algorithm based on generating a finite set of integer weighted representations. We elaborate on this in \ref{wvgdesign}.}

Nevertheless, it turns out that we can fortunately answer this question positively: there do exist exact algorithms for voting game design problems.
What follows in this section, is a study of exact algorithms for some power index voting game design problems. Of course, the most important among these problems is the variant in which we must find a weighted voting game, and output it in a weighted representation.

We approach the voting game design problem by devising an enumeration method that generates every voting game relatively efficiently.
First, we devise a ``naive'' method that enumerates all monotonic simple games in doubly exponential time (Section~\ref{mongamedesign}).
Subsequently, in Section~\ref{wvgdesign}, for the case of weighted voting games, we improve on this runtime exponentially by showing how to enumerate all weighted voting games within exponential time.
Although the runtime of this enumeration method is still exponential, we will see that it is efficient in the sense that it is polynomial in the number of games output. Moreover, the algorithm for the power index weighted voting game problem that results from this has the \textit{anytime} property: the longer we run it, the better the result becomes, and eventually the optimum solution is output.
The enumeration method is based on exploiting a new specific partial order on the class of weighted voting games.

Because we will be dealing with exponential algorithms, we make use of the $O^*$-notation: A function $f : \mathbb{R} \rightarrow \mathbb{R}$ is in $O^*(g)$ for some $g : \mathbb{R} \rightarrow \mathbb{R}$ if and only if there is a polynomial $p :\mathbb{R} \rightarrow \mathbb{R}$ such that $f \in O(g \cdot p)$. This essentially means that we make light of polynomial factors.

\subsection{Monotonic simple game design}\label{mongamedesign}
In this section we will consider the power index voting game design problem for the class of monotonic simple games $\mathcal{G}_{\mathsf{mon}}$.
There are four representation languages that can be used for monotonic simple games:
\begin{itemize}
\item $\mathcal{L}_{W}$, the winning coalition listing;
\item $\mathcal{L}_{L}$, the losing coalition listing;
\item $\mathcal{L}_{W,\min}$, the minimal winning coalition listing;
\item $\mathcal{L}_{L,\max}$, the maximal losing coalition listing.
\end{itemize}
From these languages, we obtain the following four different power index voting game design problems: $(g, \mathcal{G}_{\mathsf{mon}}, \mathcal{L}_{W})$-PVGD, $(g, \mathcal{G}_{\mathsf{mon}}, \mathcal{L}_{L})$-PVGD, $(g, \mathcal{G}_{\mathsf{mon}}, \mathcal{L}_{W, \min})$-PVGD and $(g, \mathcal{G}_{\mathsf{mon}}, \mathcal{L}_{L, \max})$-PVGD. For $g$ we can then choose any power index. Of these problems, the cases of $\mathcal{L}_{W, \min}$ and $\mathcal{L}_{L, \max}$ are the most interesting, because these languages both define the class of monotonic simple games.

We do not know of any practical situations in which this problem occurs. Therefore, we will only address this problem briefly and show for theoretical purposes that the optimal answer is computable. We do this by providing an exact algorithm.

An exact algorithm that solves $(g, \mathcal{G}_{\mathsf{mon}}, \mathcal{L}_{W, \min})$-PVGD or $(g, \mathcal{G}_{\mathsf{mon}}, \mathcal{L}_{L, \max})$-PVGD must search for the antichain of coalitions that represents the game that has a power index closest to the target power index.
This antichain of coalitions could either be a set of minimal winning coalitions, or a set of maximal losing coalitions. In either way, a simple exact algorithm for this problem would be one that considers every possible antichain, and computes for each antichain the power index for the game that the antichain represents.

Algorithm \ref{alg:monotonic} describes the process more precisely for the case that the representation language is $\mathcal{L}_{W, \min}$. We will focus on $\mathcal{L}_{W, \min}$ from now on, because the case for $\mathcal{L}_{L, \max}$ is symmetric. An algorithm for the languages $\mathcal{L}_{W}$ and $\mathcal{L}_{L}$ can be obtained by applying the transformation algorithm discussed in the last section.

\begin{algorithm}
\caption[A straightforward algorithm for solving $(g, \mathcal{G}_{\mathsf{mon}}, \mathcal{L}_{W, \min})$-PVGD.]{A straightforward algorithm for solving $(g, \mathcal{G}_{\mathsf{mon}}, \mathcal{L}_{W, \min})$-PVGD. The input is a target power index $\vec{p} = (p_1, \ldots, p_n)$. The output is an $\ell \in \mathcal{L}_{W, \min}$ such that $g(G_\ell)$ is as close as possible to $\vec{p}$.}
\label{alg:monotonic}
\begin{algorithmic}[1]
\STATE $\mathsf{bestgame} := 0$ \COMMENT{$\mathsf{bestgame}$ keeps track of
the best game that we have found, represented as a string in $\mathcal{L}_{W,
\min}$.}
\STATE $\mathsf{besterror} := \infty$ \COMMENT{$\mathsf{besterror}$ is the
error of $g(G_{\mathsf{bestgame}})$ from $\vec{p}$, according to the
sum-of-squared-errors measure.}
\FORALL{$\ell \in \mathcal{L}_{W, \min}$}
\STATE Compute $g(G_\ell) = (g(G_\ell, 1), \ldots, g(G_\ell, n))$.
\STATE $\mathsf{error} := \sum_{i=1}^n (g(G_\ell, i) - p_i)^2$.
\IF{$\mathsf{error} < \mathsf{besterror}$}
\STATE $\mathsf{bestgame} := \ell$
\STATE $\mathsf{besterror} := \mathsf{error}$
\ENDIF
\ENDFOR
\RETURN $\mathsf{bestgame}$
\end{algorithmic}
\end{algorithm}

From line 3, we see that we need to enumerate all antichains on the grand coalition. As we already said in Section \ref{cardinalities}, the number of antichains we need to enumerate is $D_n$, the $n$th Dedekind number.
Because the sequence of Dedekind numbers $(D_n)$ quickly grows very large, line 3 is what gives the algorithm a very high time complexity. The following bounds are known \cite{dedekindbounds}:
\begin{equation}\label{dedekindbounds}
2^{(1 + c'\frac{\log n}{n})E_n} \geq D_n \geq 2^{(1 + c 2^{-n/2})E_n},
\end{equation}
where $c'$ and $c$ are constants and $E_n$ is the size of the largest antichain on an $n$-set.\footnote{Korshunov devised an asymptotically equal expression \cite{korshunov}:
\[D_n \sim  2^{C(n)} e^{c(n) 2^{-n/2} + n^2 2^{-n-5} - n 2^{-n-4}}
\]
with $C(n) = \binom{n}{\lfloor n/2 \rfloor}$ and $c(n) = \binom{n}{\lfloor n/2 \rfloor - 1}$. In \cite{korshunov}, this expression is described as the number of monotonic boolean functions, which is equal to the $n$th Dedekind number.} Sperner's theorem \cite{spernerstheorem} tells us the following about $E_n$:
\begin{theorem}[Sperner's theorem]\label{spernerstheorem}
\begin{equation*}
E_n = \binom{n}{\lfloor n/2 \rfloor}.
\end{equation*}
\end{theorem}
From Sperner's theorem and Stirling's approximation, we get
\begin{equation}\label{tightboundantichain}
E_n \in \Theta \left( \frac{2^n}{\sqrt{n}} \right).
\end{equation}
We conclude that $D_n$ is doubly exponential in $n$. Algorithm \ref{alg:monotonic} therefore achieves a running time in $\Theta^*(2^{2^n} \cdot h(n))$, where $h(n)$ is the time it takes to execute one iteration of the for-loop in Algorithm \ref{alg:monotonic}. The function $h$ is an exponential function for all popular power indices (e.g., the Shapley-Shubik index and the Banzhaf index) (\cite{harisazizinfluence,dengpapadimitriou,prasadkelly1990}).

\begin{remark}
The following is an interesting related problem: Apart from the fact that $D_n$ is very large, we do not even know of an output-polynomial time procedure to enumerate all antichains on a set of $n$ elements: The simplest way to enumerate antichains would be to enumerate each possible family of coalitions, and check if that family is an antichain. Unfortunately, this method does not run in output-polynomial time.  In total, there are $2^{2^n}$ families of coalitions. Substituting the tight bound of (\ref{tightboundantichain}) into the upper bound of (\ref{dedekindbounds}), we get
\begin{equation}\label{lowerdedekindbound}
D_n \leq 2^{(1 + c' \frac{\log n}{n})k\frac{2^n}{\sqrt{n}}}
\end{equation}
for some constants $k$ and $c'$.

Exponentiating both sides of the above inequality by $\frac{\sqrt{n}}{(1 + c' \frac{\log n}{n})k}$, we see that
\begin{equation*}
D_n^{\frac{\sqrt{n}}{(1 + c' \frac{\log n}{n})k}} \leq 2^{2^n}.
\end{equation*}
This means that the number of families of subsets on an $n$-set (i.e., the right hand side of the above inequality) is super-polynomial in $n$ relative to the Dedekind number. This enumeration algorithm does thus not run in output-polynomial time, even though the $n$th Dedekind number and the number of families of coalitions on a set of $n$ elements, are both doubly-exponential. We leave this as an open problem.
\end{remark}

\subsection{Weighted voting game design}\label{wvgdesign}
Having given a simple but very slow algorithm for the PVGD-problem for the very general class of monotonic simple games, we will now see that we can do much better if we restrict the problem to smaller classes of simple games. More precisely, we will restrict ourselves to the class of weighted voting games: $\mathcal{G}_{\mathsf{wvg}}$. This class is contained in the class of linear games $\mathcal{G}_{\mathsf{lin}}$, and therefore also contained in the class of monotonic simple games $\mathcal{G}_{\mathsf{mon}}$. For this reason, we can represent a game in $\mathcal{G}_{\mathsf{wvg}}$ using any representation language that we have introduced.

As has been said in Section \ref{knownmethods}, the known literature on voting game design problems has focused on this specific variant, $(g, \mathcal{G}_{\mathsf{vwg}}, \mathcal{L}_{\mathsf{weights}})$-PVGD, with $f$ being either the Banzhaf index or the Shapley-Shubik index. Here, we will give an exact algorithm for this problem that runs in exponential time. What will turn out to make this algorithm interesting for practical purposes, is that it can be used as an anytime algorithm: we can stop execution of this algorithm at any time, but the longer we run it, the closer the answer will be to the optimum.
The advantage of this algorithm over the current local search methods is obviously that we will not get stuck in local optima, and it is guaranteed that we eventually find the optimal answer.

\subsubsection{Preliminary considerations}
Before proceeding with formally stating the algorithm, let us first address the question of which approach to take in order to find an exact algorithm for designing weighted voting games.

A possible approach to solve the $(g, \mathcal{G}_{\mathsf{vwg}}, \mathcal{L}_{\mathsf{weights}})$-PVGD problem is to use Algorithm \ref{alg:monotonic} as our basis, and check for each monotonic simple game that we find whether it is a weighted voting game. We do the latter by making use of the Hop-Skip-and-Jump algorithm that we described in Section \ref{weightedvgs}. This indeed results in an algorithm that solves the problem, but this algorithm would be highly unsatisfactory: firstly, we noted in the previous section that it is not known how to enumerate antichains efficiently. Secondly, the class of weighted voting games is a subclass of the class of monotonic simple games: in fact, we will see that there are far less weighted voting games than monotonic simple games.

The main problem we face for the $(g, \mathcal{G}_{\mathsf{vwg}}, \mathcal{L}_{\mathsf{weights}})$-PVGD problem is the fact that every weighted voting game $G \in \mathcal{G}_{\mathsf{wvg}}$ has an infinite number of weighted representations, i.e., strings in $\mathcal{L}_{\mathsf{weights}}$ that represent $G$. This is easily seen from Proposition~\ref{wvginvmult}: we can multiply the weight vector and the quota with any constant in order to obtain a new weight vector that represents the same game. On top of that, it is also possible to increase or decrease a player's weight by some amount without ``changing the game.''

Theorem 9.3.2.1 of \cite{murogathreshold} provides us with a solution to this problem, as it tells us that for every weighted voting game there exists an integer weighted representation where none of the weights nor the quota exceeds $2^{n \log n}$. This means that a possible enumeration algorithm could work by iterating over all of the $2^{(n+1)n \log n}$ integer weight vectors that have weights that fall within these bounds, and output a weight vector in case it corresponds to a weighted voting game that has not been output before. This yields an improvement over the enumeration algorithm outlined in Section \ref{mongamedesign} for the special case of weighted voting games. But we still do not consider this a satisfactory enumeration procedure because the runtime of this algorithm is still significantly larger than the known upper bounds on the number of weighted voting games. The enumeration algorithm that we propose below has a better runtime, and indeed has the property that it is also efficient in the 
sense runs in time polynomial in the number of weighted voting games it outputs. Our algorithm does not rely on weighted representations of weighted voting games; instead it works by representing weighted voting games by their sets of minimal winning coalitions.

\subsubsection{A new structural property for the class of weighted voting games}

Let us now develop the necessary theory behind the algorithm that we will propose. We will focus only on the class of \emph{canonical} weighted voting games, since for each non-canonical weighted voting game there is a canonical one that can be obtained by merely permuting the players.

The algorithm we will propose is based on a new structural property that allows us to enumerate the class of canonical weighted voting games efficiently: We will define a new relation $\preceq_{\mathsf{MWC}}$ and we will prove that for any number of players $n$ the class $\mathcal{G}_{\mathsf{cwvg}}(n)$ forms a graded poset with a least element under this relation.

\begin{definition}[$\preceq_{\mathsf{MWC}}$]
Let $G, G' \in \mathcal{G}_{\mathsf{cwvg}}(n)$ be any two canonical weighted voting games. We define $G \preceq_{\mathsf{MWC}} G'$ to hold if and only if there exists for some $k \in \mathbb{N}_{\geq 1}$ a sequence $G_1, \ldots, G_k$ of canonical weighted voting games of $n$ players, such that $G = G_1$, $G' = G_k$, and for $1 \leq i < k$ it holds that $W_{\min,i} \subset W_{\min,i+1}$ and $|W_{\min,i}| = |W_{\min,i+1}| - 1$, where $W_{\min,i}$ denotes the set of minimal winning coalitions of $G_i$.
\end{definition}

The following theorem provides the foundation for our enumeration algorithm.
\begin{theorem}\label{the:wvgposet}
For each $n$, $(\mathcal{G}_{\mathsf{cwvg}}(n), \preceq_{\mathsf{MWC}})$ is a graded poset with rank function
\begin{equation*}
\begin{array}{rclrcl}
\rho & : & \mathcal{G}_{\mathsf{cwvg}}(n) & \rightarrow & \mathbb{N} \\
 & & G & \mapsto & |W_{\min,G}|,
\end{array}
\end{equation*}
where $W_{\min,G}$ is the set of minimal winning coalitions of $G$. Moreover, the poset $(\mathcal{G}_{\mathsf{cwvg}}(n), \preceq_{\mathsf{MWC}})$ has a least element of rank $0$.
\end{theorem}

\begin{proof}[Proof of Theorem \ref{the:wvgposet}.]
By the properties of the $\preceq_{\mathsf{MWC}}$-relation, $(\mathcal{G}_{\mathsf{cwvg}}(n), \preceq_{\mathsf{MWC}})$ is a valid poset. In order to show that the poset is also \emph{graded} under the rank function $\rho$ specified in the theorem, we will prove the following lemma constructively.
\begin{lemma}\label{lem:coalitionremoval}
For every game $G \in \mathcal{G}_{\mathsf{cwvg}}(n)$ with a nonempty set $W_{\min,G}$ as its set of minimal winning coalitions, there is a coalition $C \in W_{\min,G}$ and a game $G' \in \mathcal{G}_{\mathsf{cwvg}}(n)$ so that $W_{\min,G} \setminus \{C\}$ is the set of minimal winning coalitions of $G'$.
\end{lemma}
From Lemma \ref{lem:coalitionremoval}, the claim follows: It implies that there is a unique minimal element, which is the game with no minimal winning coalitions. This implies that the poset has a least element, and it implies that condition \emph{(i.)} in Definition \ref{def:posets} holds. Condition \emph{(ii.)} and \emph{(iii.)} follow immediately from the definitions of $\preceq_{\mathsf{MWC}}$ and $\rho$.

To prove Lemma \ref{lem:coalitionremoval}, we first prove the following two preliminary lemmas (\ref{lem:weightchange} and \ref{lem:weightsalldiffer}).

\begin{lemma}\label{lem:weightchange}
Let $G = (N = \{1, \ldots, n\}, v)$ be a weighted voting game, and let $\ell = [q; w_1, \ldots, w_n]$ be a weighted representation for $G$. For each player $i$ there exists an $\epsilon > 0$ such that for all $\epsilon' < \epsilon$, the vector $\ell' = [q; w_1, \ldots, w_i + \epsilon', w_{i+1}, \ldots, w_n]$ is also a weighted representation for $G$.
\end{lemma}
Informally, this lemma states that it is always possible to increase the weight of a player by some amount without changing the game.
\begin{proof}
Let $w_{\max} = \max\{C \subseteq N : v(C) = 0\}$ and $w_{\min} = \min\{C \subseteq N : v(C) = 1\}$.
Take $\epsilon = w_{\min} - w_{\max}$ and note that $\epsilon > 0$. Increasing any player's weight by any positive amount $\epsilon'$ that is less than $\epsilon$ does not turn any losing coalition in a winning coalition. Obviously, as $\epsilon' > 0$, this change of weight also does not turn any winning coalition into a losing coalition.
\end{proof}

\begin{lemma}\label{lem:weightsalldiffer}
Let $G = (N = \{1, \ldots, n\}, v)$ be a weighted voting game. There exists a weighted representation $\ell$ for $G$ such that for all $C, C' \in 2^N$, $C \not= C'$, for which $v(C) = v(C') = 1$, it holds that $w_\ell(C) \not= w_\ell(C')$.
\end{lemma}
Or, informally stated: for every weighted voting game, there exists a weighted representation such that all winning coalitions have a different weight.
\begin{proof}
Let $\ell' = [q; w_1, \ldots, w_n]$ be a weighted representation for $G$. Fix an arbitrary player $i$. By Lemma $\ref{lem:weightchange}$, there is an $\epsilon > 0$ such that increasing $w_i$ by any value $\epsilon' \in (0, \epsilon)$ will result in another weighted representation for $G$.
Let $\mathcal{E}$ be the set of choices for $\epsilon'$ such that increasing $w_i$ by $\epsilon'$ yields a weighted representation $\ell$ where there are two coalitions $C, C' \in 2^{N}$, $i \in C$, $i \not\in C'$, that have the same weight under $\ell$. There are finitely many such pairs $(C,C')$ so $\mathcal{E}$ is finite and therefore $(0, \epsilon) \backslash \mathcal{E}$ is non-empty. By picking for $\epsilon'$ any value in $(0, \epsilon) \backslash \mathcal{E}$, and increasing the weight of player $i$ by $\epsilon'$, we thus end up with a weighting $\ell$ in which there is no coalition $C$ containing $i$ such that $w_\ell(C)$ is equal to any coalition not containing $i$. Furthermore, if $C, C'$ are two arbitrary coalitions that have distinct weight under $\ell'$, then certainly they will have distinct weight under $\ell$.

Therefore, by sequentially applying the above operation for all $i \in N$, we end up with a weighting $\ell$ for which it holds for every player $i$ that there is no coalition $C$ containing $i$ such that $w_{\ell}(C)$ is equal to the weight of any coalition not containing $i$. This implies that all coalitions have a different weight under $\ell$, and completes the proof.
\end{proof}

Using Lemma \ref{lem:weightsalldiffer}, we can prove Lemma \ref{lem:coalitionremoval}, which establishes Theorem \ref{the:wvgposet}.
\begin{proof}[Proof of Lemma~\ref{lem:coalitionremoval}]
Let $G = (\{1, \ldots, n\}, v)$ be a canonical weighted voting game. Let $W_{\min,G}$ be its set of minimal winning coalitions and let $\ell = [q; w_1, \ldots, w_n]$ be a weighted representation for which it holds that all winning coalitions have a different weight. By Lemma \ref{lem:weightsalldiffer}, such a representation exists. We will construct an $\ell''$ from $\ell$ for which it holds that it is a weighted representation of a canonical weighted voting game with $W_{\min,G} \setminus\{C\}$ as its list of minimal winning coalitions, for some $C \in W_{\min,G}$.

Let $i$ be the highest-numbered player that is in a coalition in $W_{\min,G}$, i.e., $i$ is a least desirable nondummy player.
We may assume without loss of generality that $w_j = 0$ for all $j > i$.
Let $C \in W_{\min,G}$ be the minimal winning coalition containing $i$ 
with the lowest weight among all MWCs in $W_{\min,G}$ that contain $i$.

Next, define $\ell'$ as $[q; w_1, \ldots, w_i - (w_\ell(C) - q), \ldots, w_n]$. Note that under $\ell'$ the weights of the players are still decreasing and nonnegative: we have that $w_{\ell}(C \backslash\{i\}) < q$ (due to $C \backslash\{i\}$ being a losing coalition, because $C$ is a MWC). This implies  $w_{\ell}(C \backslash\{i\}) + w_i = w_{\ell}(C) < q + w_i$, and is equivalent to $w_i > w_{\ell}(C) - q$ (i.e., the new weight of player $i$ is indeed nonnegative under $\ell'$. Player $i$'s weight is decreased in $\ell'$, but not by enough to make losing even the lightest of all MWCs that contain $i$, so now $G_{\ell'} = G_\ell = G$ and $w_{\ell'}(C) = q$. Moreover, the weights of the coalitions in $W_{\min,G}$ that contain player $i$ are still mutually distinct under $\ell'$.

We now decrease player $i$'s weight by an amount that is so small that the only minimal winning coalition that turns into a losing coalition is $C$. Note that under $\ell'$, minimal winning coalition $C$ is still the lightest minimal winning coalition containing $i$. Let $C' \in W_{\min}$ be the second-lightest minimal winning coalition containing $i$. Obtain $\ell''$ by decreasing $i$'s weight (according to $\ell'$) by a positive amount smaller than $w_{\ell'}(C') - w_{\ell'}(C)$. Coalition $C$ will become a losing coalition and all other minimal winning coalitions will stay winning. No new minimal winning coalition is introduced in this process: Suppose there would be such a new minimal winning coalition $S$, then $S$ contains only players that are at least as desirable as $i$ (the other players have weight $0$). In case $i \not\in S$, we have that $S$ would also be a minimal winning coalition in the original game $G$ because $w_{\ell''}(S) = w_{\ell}(S) \geq q$, which is a contradiction. In case $i \in S$,
 it must be that $S \backslash \{i\}$ is winning in the original game $G$ (as $i$ was picked to be a least desirable nondummy player). Thus, $w_{\ell''}(S \backslash \{i\} = w_{\ell} (S \backslash \{i\}) \geq q$, and this is in contradiction with $S$ being an MWC in $G_{\ell''}$.

$G_{\ell''}$ is therefore an $n$-player canonical weighted voting game whose set of MWCs form a subset of the MWCs of $G$, and the cardinalities of these two sets  differ by $1$. This proves the claim.
\end{proof}
\end{proof}

In Figure \ref{fourplayerposet}, $(\mathcal{G}_{\mathsf{cwvg}}(4), \preceq_{\mathsf{MWC}})$ is depicted graphically. Note that this is \textit{not} precisely the Hasse diagram of the poset $(\mathcal{G}_{\mathsf{cwvg}}(4), \preceq_{\mathsf{MWC}})$ (see the explanation in the caption of this figure; the reason that we do not give the Hasse diagram is because the Hasse diagram is not a very convenient way of representing $(\mathcal{G}_{\mathsf{cwvg}}(4), \preceq_{\mathsf{MWC}})$).

\begin{figure}
\centering
\includegraphics[scale=0.7,angle=0]{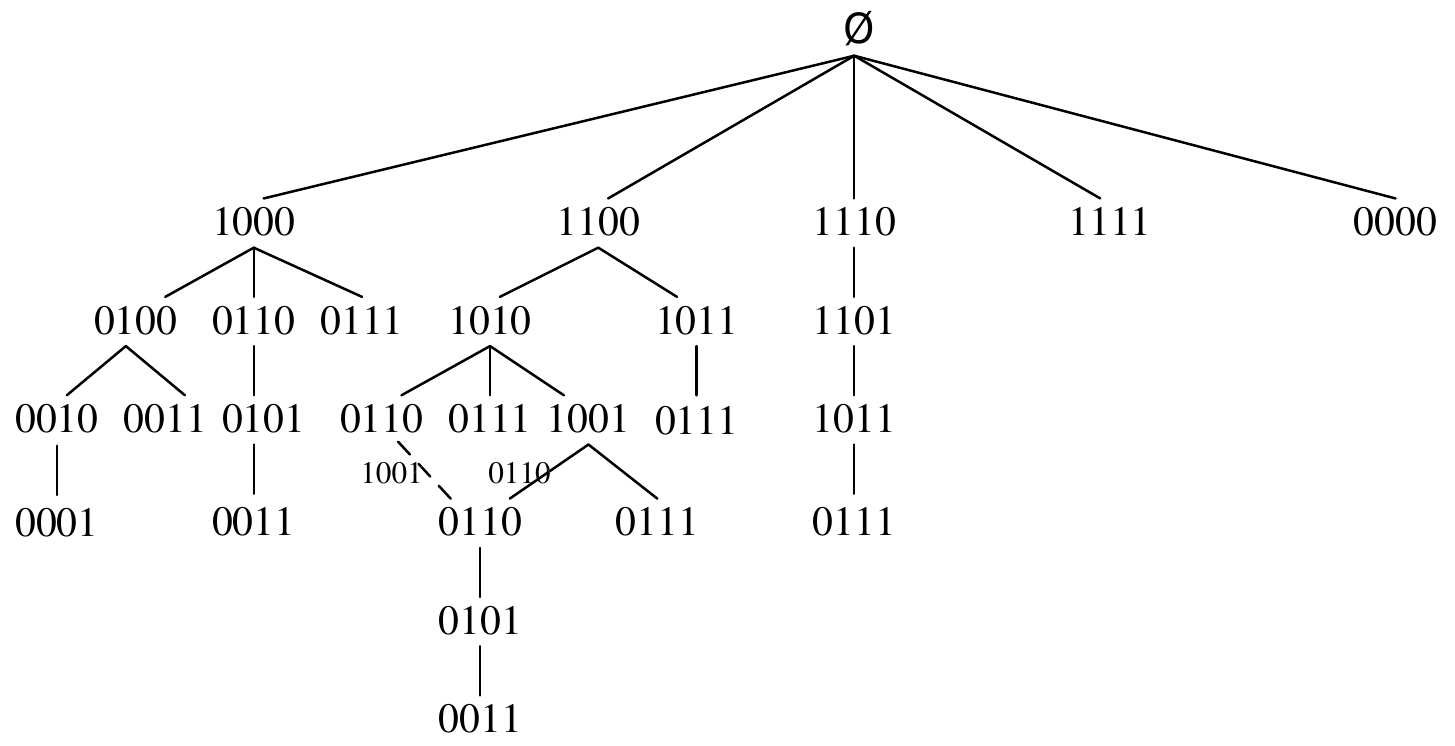}
\caption[Graphical depiction of $(\mathcal{G}_{\mathsf{cwvg}}(4), \preceq_{\mathsf{MWC}})$.]{
Graphical depiction of $(\mathcal{G}_{\mathsf{cwvg}}(4), \preceq_{\mathsf{MWC}})$. Each node in this graph represents a canonical weighted voting game of four players. It should be read as follows: each node has the characteristic vector of a minimal winning coalition as a label. The set of minimal winning coalitions of a game that corresponds to a certain node $n$ in the graph, are those coalitions that are described by the set $V_n$ of vectors  that is obtained by traversing the path from the top node to $n$ along the solid edges. The top node corresponds to the canonical weighted voting game with zero minimal winning coalitions (i.e., every coalition loses). The actual Hasse diagram of this poset can be obtained by changing the label of each node $n$ to $V_n$ and including the solid edges as well as the dashed edge in the diagram.}
\label{fourplayerposet}
\end{figure}
\normalsize

Next, we show that $(\mathcal{G}_{\mathsf{cwvg}}(n), \preceq_{\mathsf{MWC}})$ is not a tree for $n \geq 4$. When we will state our algorithm in the next section, it will turn out that this fact makes things significantly more complicated.
\begin{proposition}\label{the:notatree}
For $n \geq 4$, $(\mathcal{G}_{\mathsf{cwvg}}(n), \preceq_{\mathsf{MWC}})$ is not a tree.
\end{proposition}
\begin{proof}
We will give an example of a game in $(\mathcal{G}_{\mathsf{cwvg}}(4), \preceq_{\mathsf{MWC}})$ that covers multiple games.\footnote{In fact, from inspecting Figure \ref{fourplayerposet} and the explanation given in its caption, it may already be rather obvious to the reader which example game we intend.} A similar example for $n > 4$ is obtained by adding dummy players to the example that we give here.

Consider the following weighted representation of a canonical weighted voting game over players $\{1,2,3,4\}$:
\begin{equation*}
\ell = [4; 3, 2, 2, 1].
\end{equation*}
The set of characteristic vectors $C_{\min, \ell}$ of minimal winning coalitions of $G_{\ell}$ is as follows:
\begin{eqnarray*}
C_{\min, \ell} & = & \{1100, 1010, 0110, 1001\}.
\end{eqnarray*}
Next, consider the weighted voting games $\ell'$ and $\ell''$:
\begin{eqnarray*}
\ell' & = & [4; 3, 1, 1, 1] \\
\ell'' & = & [2; 1, 1, 1, 0],
\end{eqnarray*}
with respectively the following sets of characteristic vectors of minimal winning coalitions:
\begin{eqnarray*}
C_{\min, \ell'} & = & \{1100, 1010, 1001\}, \\
C_{\min, \ell''} & = & \{1100, 1010, 0110\}.
\end{eqnarray*}
It can be seen that $C_{\min, \ell'} = C_{\min, \ell} \setminus \{0110\}$ and $C_{\min, \ell''} = C_{\min, \ell} \setminus \{1001\}$.
\end{proof}

\subsubsection{The algorithm}\label{sec:thealgorithm}
We will use the results from the previous section to develop an exponential-time exact algorithm for $(f, \mathcal{G}_{\mathsf{cwvg}}, \mathcal{L}_{W,\min})$-PVGD, and also for $(f, \mathcal{G}_{\mathsf{cwvg}}, \mathcal{L}_{\mathsf{weights}})$-PVGD. The way this algorithm works is very straightforward: Just as in algorithm \ref{alg:monotonic}, we enumerate the complete class of games (weighted voting games in this case), and we compute for each game (that is output by the enumeration algorithm) the distance from the target power index.

Recall that the problem with Algorithm \ref{alg:monotonic} was that the enumeration procedure is not efficient. For the restriction to weighted voting games, we are able to make the enumeration procedure more efficient. We will use Theorem \ref{the:wvgposet} for this: The key is that it is possible to generate the minimal winning coalition listing of canonical weighted games of rank $i$ fairly efficiently from the minimal winning coalition listing of canonical weighted voting games of rank $i-1$.

The following theorem shows us how to do this. To state this theorem, we will first generalize the truncation-operation from Definition \ref{coalitionoperations}.

\begin{definition}[Right-truncation]
Let $S \subseteq N$ be a coalition on players $N = \{1,\ldots, n\}$.
The \textit{$i$th right-truncation} of $S$, denoted $\mathsf{rtrunc}(S,i)$, is defined as
\begin{equation*}
\mathsf{rtrunc}(S,i) = \begin{cases} S \setminus \{P(S,i), \ldots, n\} \text{ if } 0 < i \leq |S| , \\ S \text{ if } i = 0, \\ \text{undefined otherwise,} \end{cases}
\end{equation*}
where $P(S,i)$ is the $i$th highest-numbered player among the players in $S$.
\end{definition}
In effect, the $i$th right-truncation of a coalition $S$ (for $i\leq|S|$) is the coalition that remains when the $i$ highest-numbered players are removed from $S$.

\begin{theorem}\label{the:findingnewcoalitions}
For any $n$, let $G, G' \in \mathcal{G}_{\mathsf{clin}}(n)$ be a pair of canonical linear games with respective sets of minimal winning coalitions $W_{\min}$ and $W_{\min}'$, such that $W_{\min} \subseteq W_{\min}'$. Let $L_{\max}$ and $L_{\max}'$ be the sets of maximal losing coalitions of $G$ and $G'$ respectively. There is a $C \in L_{\max}$ and an $i \in \mathbb{N}$ with $0 \leq i \leq n$ such that $W_{\min}' = W_{\min} \cup \{\mathsf{rtrunc}(C,i)\}$.
\end{theorem}
\begin{proof}
Because $G$ is covered by $G'$, by definition there is a coalition $C' \not\in W_{\min}$ such that $W_{\min}' = W_{\min} \cup \{C'\}$. Coalition $C'$ cannot be a superset of a coalition in $W_{\min}$ because then it would not be a \emph{minimal} winning coalition in $G'$. Therefore, $C'$ is a losing coalition in $G$, so it must be a subset of a coalition in $L_{\max}$. Suppose for contradiction that $C'$ is not a right-truncation of a maximal losing coalition $C \in L_{\max}$. So there is a $C \in L_{\max}$ such that $C'$ is a subset of $C$, but not a right-truncation of $C$. This means that in $C'$, some player $j$ from $C$ is not present, while at least one less desirable player $k>j$ from $C$ is in $C'$. This implies that there is a left-shift $C''$ of $C'$ such that $C''$ is a subset of a coalition in $L_{\max}$: $C''$ is obtained from $C'$ by replacing $j$ by $k$. $C''$ is a subset of $C$, so $C''$ is still a losing coalition in $G$. $C''$ is thus not a superset of any coalition in $W_{\min}$, and hence
$C''$ is also not a superset of any coalition in $W_{\min}'$. So \emph{$C''$ is a losing coalition in $G'$}. But $G'$ is a canonical linear game, so by the desirability relation $1 \succeq_D \cdots \succeq_D n$, \emph{$C''$ is a winning coalition in $G'$} because it is a left-shift of the winning coalition $C'$. This is a contradiction.
\end{proof}

From Theorem \ref{the:findingnewcoalitions}, it becomes apparent how to use $(\mathcal{G}_{\mathsf{cwvg}}(n), \preceq_{\mathsf{MWC}})$ for enumerating the class of $n$-player canonical weighted voting games. We start by outputting the $n$-player weighted voting game with zero minimal winning coalitions. After that, we repeat the following process: generate the $\mathcal{L}_{W,\min}$-representation of all canonical weighted voting games with $i$ minimal winning coalitions, using the set of canonical weighted voting games games with $i-1$ minimal winning coalitions (also represented in $\mathcal{L}_{W,\min}$). Once generated, we have the choice to output the games in their $\mathcal{L}_{W,\min}$-representation or in their $\mathcal{L}_{\mathsf{weights}}$-representation, by using the Hop-Skip-and-Jump algorithm presented in Section \ref{actualhopskipjumpalgorithm}.

Generating the set of games of $i$ minimal winning coalitions works as follows: For each game of $i-1$ minimal winning coalitions, we obtain the set of maximal losing coalitions by using the Hop-Skip-and-Jump algorithm. Next, we check for each maximal losing coalition $C$ whether there is a right-truncation of $C$ that we can add to the set of minimal winning coalitions, such that the resulting set represents a weighted voting game. Again, testing whether a game is a weighted voting game is done by using the Hop-Skip-and-Jump algorithm. If a game turns out to be weighted, we can store it and output it.

There is one remaining problem with this approach: It outputs duplicate games. If $(\mathcal{G}_{\mathsf{cwvg}}(n), \preceq_{\mathsf{MWC}})$ were a tree, then this would not be the case, but by Proposition \ref{the:notatree} it is not a tree for any $n \geq 4$. Therefore, we have to do a \textit{duplicates-check} for each weighted voting game that we find. In principle, this seems not to be so difficult: For each game that we find, sort its list of minimal winning coalitions, and check if this list of coalitions already occurs in the array of listings of minimal winning coalitions that correspond to games that we already found. The problem with this is that the list can grow very large, so these checks are then very time- and space-consuming operations.

We will therefore use a different method for doing this ``duplicates-check''. Suppose that we have found an $n$-player canonical weighted voting game $G$ of $i$ minimal winning coalitions by adding a coalition $C$ to a minimal winning coalition listing of a canonical weighted voting game that we have already found. We first sort $G$'s list of minimal winning coalitions. After that, we check for each coalition $C'$ that occurs before $C$ in this sorted list, whether $C'$'s removal from the list results in a list of minimal winning coalitions of a canonical weighted voting game. If there is such a $C'$, then we discard $G$, and otherwise we keep it. This way, it is certain that each canonical weighted voting game will be generated only once.

Algorithm \ref{alg:wvgoutput} gives the pseudocode for this enumeration method. The array element $\mathsf{games}[i]$ will be the list of canonical weighted voting games that have $i$ minimal winning coalitions. The value of $i$ can not exceed $\binom{n}{\lfloor n/2 \rfloor}$ by Theorem $\ref{spernerstheorem}$. The games are represented in language $\mathcal{L}_{W,\min}$. The algorithm iterates from every new game found, starting from the game in $\mathsf{games}[0]$, which is the $n$-player canonical weighted voting game with zero minimal winning coalitions.
\begin{algorithm}
\caption{An enumeration algorithm for the class of $n$ player canonical weighted voting games. $\mathsf{hopskipjump}$ refers to the Hop-Skip-and-Jump algorithm.}
\label{alg:wvgoutput}
\begin{algorithmic}[1]
\STATE \textbf{Output} $[1; 0,\ldots, 0]$
\STATE $\mathsf{games}[0] := \{ \varnothing \}$
\FOR{$i := 1$ to $\binom{n}{\lfloor n/2 \rfloor}$}
\FORALL{$W_{\min} \in \mathsf{games}[i-1]$} \label{outerloop1}
\STATE \COMMENT{Obtain the maximal losing coalitions:} \label{outerloop2}
\STATE $L_{\max} := \mathsf{hopskipjump}(W_{\min})$\label{linehopskipjump}
\FORALL{$C \in L_{\max}$}
\FOR{$j := 0$ to $n$}
\IF{$\mathsf{isweighted}(W_{\min} \cup \mathsf{rtrunc}(C,j))$}\label{checkweighted}
\IF{$W_{\min} \cup \mathsf{rtrunc}(C,j)$ passes the duplicates check (see discussion above)}\label{dupcheck}
\STATE \textbf{Output} a weighted representation of the voting game with minimal winning coalitions $W_{\min} \cup \mathsf{rtrunc}(C,j))$.\label{remainder1}
\STATE \textbf{Append} $W_{\min} \cup \mathsf{rtrunc}(C,j))$ to $\mathsf{games}[i]$.\label{remainder2}
\ENDIF
\ENDIF \label{checkweightedend}
\ENDFOR
\ENDFOR \label{outerloopend}
\ENDFOR
\ENDFOR
\end{algorithmic}
\end{algorithm}
Correctness of the algorithm follows from our discussion above. We will now analyze the time-complexity of the algorithm.
\begin{theorem}
Algorithm \ref{alg:wvgoutput} runs in $O^*(2^{n^2+2n})$ time.
\end{theorem}
\begin{proof}
Lines 5 to 16 are executed at most once for every canonical weighted voting game. From Sperner's theorem, we know that any list of minimal winning coalitions has fewer than $\binom{n}{\lfloor n/2 \rfloor}$ elements. So by the runtime of the Hop-Skip-and-Jump algorithm, line~\ref{linehopskipjump} runs in time $O\left(n {\binom{n}{\lfloor n/2 \rfloor}}^2 + n^3 \binom{n}{\lfloor n/2 \rfloor}\right) = O(n^2 \sqrt{n} 2^{\bk{2}n})$. Within an iteration of the outer loop (line \ref{outerloop1}), lines \ref{checkweighted} to \ref{checkweightedend} are executed at most $n \binom{n}{\lfloor n/2 \rfloor} = O(\sqrt{n} 2^n)$ times (because $L_{\max}$ is also an antichain, Sperner's theorem also applies for maximal losing coalitions). The time-complexity of one execution of lines \ref{checkweighted} to \ref{checkweightedend} is as follows.
\begin{itemize}
\item At line \ref{checkweighted}, we must solve a linear program, taking time $O\left(n^{4.5} \binom{n}{\lfloor n/2 \rfloor}\right) = O(n^4 2^n)$ using Karmarkar's interior point algorithm \cite{karmarkar}.
\item At line \ref{dupcheck}, we must execute the duplicates check. This consists of checking for at most $\binom{n}{\lfloor n/2 \rfloor}$ sets of minimal winning coalitions whether they are weighted. This involves running the Hop-Skip-and-Jump algorithm, followed by solving a linear program. In total, this takes $O(n^{3} \sqrt{n} 2^{2n})$ \tk{time}.
\item Lines \ref{remainder1} and \ref{remainder2} take linear time.
\end{itemize}
Bringing everything together, we see that a single pass from lines \ref{outerloop2} to \ref{outerloopend} costs us $O(n^{4} 2^{3n})$ time.
As mentioned earlier, these lines are executed at most $|\mathcal{G}_{\mathsf{cwvg}}(n)|$ times.
We know that $|\mathcal{G}_{\mathsf{wvg}}(n)| \in O(2^{n^2-n})$ (from Corollary~\ref{no-wvgs} in the previous section), and of course $|\mathcal{G}_{\mathsf{cwvg}}(n)| < |\mathcal{G}_{\mathsf{wvg}}(n)|$, hence lines \ref{outerloop2} to \ref{outerloopend} are executed at most $O(2^{n^2-n})$ times, and therefore the runtime of the algorithm is $O(2^{n^2+2n} n^{4}) = O^*(2^{n^2+2n})$.
\end{proof}
Although the runtime analysis of this algorithm that we gave is not very precise, the main point of interest that we want to emphasize is that this method runs in exponential time, instead of doubly exponential time. We can also show that this algorithm runs in an amount of time that is only polynomially greater than the amount of data output. This implies that Algorithm \ref{alg:wvgoutput} is essentially the fastest possible enumeration algorithm for canonical weighted voting games, up to a polynomial factor.

\begin{theorem}
Algorithm \ref{alg:wvgoutput} runs in output-polynomial time, i.e., a polynomial in the number of bits that Algorithm \ref{alg:wvgoutput} outputs.
\end{theorem}
\begin{proof}
Lines \ref{outerloop2} to \ref{outerloopend} are executed less than $|\mathcal{G}_{\mathsf{cwvg}}(n)|$ times.
From (\ref{no-cwvgs}), we have as a lower bound that $|\mathcal{G}_{\mathsf{cwvg}}(n)| \in \Omega(2^{n^2(1 - \frac{10}{\log n})} / n!2^n)$.
One execution of lines \ref{outerloop2} to \ref{outerloopend} costs $O(n^{4} 2^{3n})$ time, and thus one iteration \bk{runs in
\begin{equation*}
O(n^{4} 2^{3n}) \subseteq O\left(2^{n^2(1 - \frac{10}{\log n})} / n!2^n\right) \subseteq O(|\mathcal{G}_{\mathsf{cwvg}}(n)|)
\end{equation*}
time.} We conclude that the algorithm runs in $O(|\mathcal{G}_{\mathsf{cwvg}}(n)|^2)$ time.
\end{proof}

\begin{remark}\label{depthfirst}
We can not give a very sharp bound on the space complexity of Algorithm \ref{alg:wvgoutput}, because we do not know anything about the maximum cardinality of an antichain in $(\mathcal{G}_{\mathsf{cwvg}}(n), \preceq_{\mathsf{MWC}})$. However, it can be seen that it is also possible to generate the games in this poset in a depth-first manner, instead of a breadth-first manner like we do now. In that case, the amount of space that needs to be used is bounded by the maximum length of a chain in $(\mathcal{G}_{\mathsf{cwvg}}(n), \preceq_{\mathsf{MWC}})$. This is a total amount of $O(\frac{2^n}{\sqrt{n}})$ space.
\end{remark}

Now that we have this enumeration algorithm for weighted voting games, we can use the same approach as in algorithm $\ref{alg:monotonic}$ in order to solve the $(f,\mathcal{G}_{\mathsf{cwvg}},\mathcal{L}_{\mathsf{weights}})$-PVGD problem: for each game that is output, we simply compute the power index of that game and check if it is closer to the optimum than the best game we have found up until that point.

\subsection{Improvements and optimizations}\label{improvingthealgorithm}
Algorithm \ref{alg:wvgoutput} is in its current state not that suitable for solving the $(f,\mathcal{G}_{\mathsf{cwvg}},\mathcal{L}_{\mathsf{weights}})$-PVGD problem in practice. In this section we will make several improvements to the algorithm. This results in a version of the enumeration algorithm of which we expect that it outputs canonical weighted voting games at a steady rate. We will see that this gives us a practically applicable anytime-algorithm for the $(\beta,\mathcal{G}_{\mathsf{cwvg}},\mathcal{L}_{\mathsf{weights}})$-PVGD problem for small numbers of players.

Section \ref{sec:improvedlp} shows how we can make the system of linear inequalities (\ref{systemofinequalities}) smaller.
In Section \ref{sec:findingcoalitions}, we will improve Theorem \ref{the:findingnewcoalitions} in order to more quickly find new potential minimal winning coalitions to extend our weighted voting games with. Lastly, in Section \ref{sec:outputpolynomial} we give an output-polynomial time algorithm for enumerating all ceiling coalitions, given a set of roof coalitions.

It is important to note that these three improvements combined eliminate the need to generate the complete list of maximal losing coalitions of the weighted voting games that we enumerate. Instead, it suffices to only keep track of the sets of minimal winning coalitions and ceiling coalitions.

\subsubsection{An improved linear program for finding the weight vector of a weighted voting game}\label{sec:improvedlp}
When finding a weight vector for a weighted voting game of which we obtained the minimal winning coalitions and maximal losing coalitions, we proposed in the previous section to do this by solving the system of inequalities (\ref{systemofinequalities}). In \cite{hopskipjump} it is noted that we can make this system much more compact, as follows.

First of all we can reduce the number of inequalities in our system by observing that a minimal winning coalition $C$ which is not a roof, always has a higher total weight than at least one roof, in a canonical weighted voting game. This is because $C$ is a superset of a left-shift of some roof. In the same way, a maximal losing coalition which is not a ceiling, always has a lower total weight than at least one ceiling.
Therefore, adding the inequalities $w_1 \geq \cdots \geq w_n$ to our system of inequalities (\ref{systemofinequalities}) allows us to remove a lot of other inequalities from (\ref{systemofinequalities}), because it now suffices to only make sure that out of all minimal winning coalitions, only the roofs have a higher weight than $q$; and out of all maximal losing coalitions, only the ceilings have a lower total weight than $q$.

Secondly, we can reduce the number of variables (weights) in (\ref{systemofinequalities}) by noting that if two players $i$ and $i+1$ are equally desirable, then $w_i = w_{i+1}$. Therefore, we need only one representative variable from each set $D$ of players for which it holds that that
\begin{enumerate}
\item the players in $D$ are pairwise equally desirable, and
\item any player in $N \setminus D$ is strictly less or strictly more desirable than a player in $D$.
\end{enumerate}
By reducing the number of inequalities and variables in this way, we can in most cases drastically decrease the time it takes to find a solution to (\ref{systemofinequalities}).

\subsubsection{A better way of finding new minimal winning coalitions}\label{sec:findingcoalitions}
Theorem \ref{the:findingnewcoalitions} allows us to find potential minimal winning coalitions that we can extend our weighted voting games with. We will now see that we do not really need to consider every right-truncation of every maximal losing coalition: In fact, we only need to look at ceiling coalitions.

\begin{theorem}\label{the:findingnewcoalitions2}
For any $n$, let $G, G' \in \mathcal{G}_{\mathsf{wvg}}(n)$ be a pair of weighted voting games such that $G$ is covered by $G'$ in $(\mathcal{G}_{\mathsf{cwvg}}(n), \preceq_{\mathsf{MWC}})$. Let $W_{\min, G}$ and $W_{\min, G'}$ be the sets of minimal winning coalitions of $G$ and $G'$ respectively, and let $L_{\mathsf{ceil},G}$ and $L_{\mathsf{ceil},G'}$ be the sets of ceiling coalitions of $G$ and $G'$ respectively. There is a $C \in L_{\mathsf{ceil},G}$ and an $i \in \mathbb{N}$ with $0 \leq i \leq n$ such that $W_{\min, G'} = W_{\min, G} \cup \mathsf{rtrunc}(C,i)$.
\end{theorem}
\begin{proof}
Let $W_{\min,G}$ and $W_{\min,G'}$ be the sets of minimal winning coalitions of games $G$ and $G'$ respectively.
Because $G$ is covered by $G'$, by definition there is a coalition $C \not\in W_{\min, G}$ such that $W_{\min, G'} = W_{\min, G} \cup C$.
By Theorem \ref{the:findingnewcoalitions}, $C$ is a right-truncation of a coalition in $L_{\max,G}$.
Suppose for contradiction that $C$ is not a right-truncation of a ceiling in $L_{\mathsf{ceil},G}$.
Then there is a ceiling $C' \in L_{\mathsf{ceil},G}$ such that $C$ is a subset of a right-shift of $C'$, and there is a left-shift $C''$ of $C$, $C'' \not=C$, such that $C''$ is also a subset of a right-shift of $C'$. Coalition $C''$ is not a superset of $W_{\min, G}$ because $C'$ is losing in $G$, and $C''$ is not a superset of $C$ either, because $C''$ is a left-shift of $C$ and is unequal to $C$. So it follows that \emph{$C''$ is a losing coalition in $G'$}.

But $G'$ is a canonical weighted voting game, so the desirability relation $1 \succeq_D \cdots \succeq_D n$ is satisfied. Because $C''$ is a left-shift of $C$, and $C$ is winning in $G'$, it follows that \emph{$C''$ is a winning coalition in $G'$}. This is a contradiction.
\end{proof}

\subsubsection{An algorithm for obtaining the ceiling-list of a canonical weighted voting game}\label{sec:outputpolynomial}
In Section \ref{roofvgs}, we derived that the $(\mathcal{L}_{\mathsf{roof}}, \mathcal{L}_{\mathsf{ceil}})$-VGS problem does not have a polynomial time algorithm because the output may be exponentially sized in the input. Moreover, as proved by Polym\'{e}ris and Riquelme \cite{polymeris13}, an output-polynomial time algorithm for this problem would have sensational consequences, as it would imply a polynomial time algorithm for the \emph{monotone boolean duality} problem \cite{fredman96}: a well-known problem that can be solved in sub-exponential time, of which it is not known whether it admits a polynomial time algorithm.

Finding an output-polynomial algorithm for $(\mathcal{L}_{\mathsf{roof}}, \mathcal{L}_{\mathsf{ceil}})$-VGS is thus a very interesting open problem, but due to its alleged difficulty\footnote{We thank Andreas Polym\'{e}ris and Fabi\'{a}n Riquelme for pointing out to us the connection to the monotone boolean duality problem, as well as for pointing out an error in a preliminary version of this papers.} we instead resort to studying the $(\mathcal{L}_{W, \min}, \mathcal{L}_{\mathsf{ceil}})$-VGS problem.

Of course, one could simply solve the latter problem by using the Hop-Skip-and-Jump algorithm, described in Section \ref{weightedvgs}. This would provide us with a list of MLCs of the input game, after which we could filter out the ceilings. This algorithm (i.e., filtering the shelters from the MWCs and running the Hop-Skip-and-Jump algorithm on those shelters) would run in $O(nm^2 + n^3 m)$ time, where $m$ is the number of MWCs. Below, we will provide an alternative algorithm for the special case where we only need to output the ceilings of a given game. We will use in the remainder of this section some notions from Definition \ref{coalitionoperations}.


\begin{theorem}
Let $G \in \mathcal{G}_{\mathsf{clin}}(n)$ be a canonical linear game on players $N = \{1, \ldots, n\}$, let $W_{\min}$ be the set of MWCs of $G$, let $\mathcal{C}$ be the set of ceilings of $G$, and let $C \in \mathcal{C}$ such that $a(C) > 0$. Then there exists $i \in \mathbb{N}_{\geq 0}, 1 < i \leq |C| - |\mathsf{trunc}(C)|$, such that $\mathsf{trunc}(C) \cup \{a(C)\} \cup \{a(C) + j : 2 \leq j \leq i\}$ is a minimal winning coalition.
\end{theorem}
\begin{proof}
$C$ is losing, so $\mathsf{trunc}(C)$ is losing, and $\mathsf{trunc}(C) \cup \{b(C)\} \cup \{b(C) + j : 2 \leq j \leq |C| - |\mathsf{trunc}(C)|\}$ is a left-shift of $C$, and hence winning. Therefore there exists $i \in \mathbb{N}_{\geq 0}, 2 \leq i \leq |C| - |\mathsf{trunc}(C)|$ such that
\begin{itemize}
 \item $C' := \mathsf{trunc}(C) \cup \{a(C)\} \cup \{a(C) + j : 0 \leq j \leq i\}$ is winning, and
 \item ($C'' := \mathsf{trunc}(C) \cup \{a(C)\} \cup \{a(C) + j : 2 \leq j \leq i-1\}$ is losing or $C'' = C'$).
\end{itemize}
By canonicity, this means that $C'$ is a minimal winning coalition.
\end{proof}

From the above theorem it becomes clear how to generate efficiently the set of ceilings from the set of minimal winning coalitions: For each MWC $S$, it suffices to check for all $k \in \mathbb{N}_{\geq 0}, k \leq n - b(S)$ whether:
\begin{itemize}
 \item $(S \backslash \{a(S)-1\}) \cup \{a(S)\} \cup \{b(S) + j : 1 \leq j \leq k\}$ is a ceiling (in case $a(S) - 1 \in S$),
 \item $(S \backslash \{b(S)\}) \cup \{b(S) + j : 1 \leq j \leq k\}$ is a ceiling.
\end{itemize}
This would generate all ceilings $C$ with the property that $b(C) > 0$. There are furthermore at most $n$ ceilings $C$ for which it holds that $b(C) = 0$, and it is clear that such coalitions can be generated and checked straightforwardly.

The runtime of this implied algorithm is theoretically no better than the Hop-Skip-and-Jump algorithm. However, due to the simplicity of this algorithm, and due to the fact that this algorithm only finds the ceilings (of which there are in general much less than that there are MWCs), we expect this algorithm to run much faster in practice, in most cases.

\section{Experiments}\label{sec:experiments}

Before we turn to the results of some relatively large scale experiments, let us visualize the results for just $n=3$ players, because we can easily depict these in two dimensions. Figure~\ref{fig:simplex} shows the 3-player simplex, with the vertices labeled by the player numbers.
\begin{figure}
\center
\includegraphics[width=.8\textwidth]{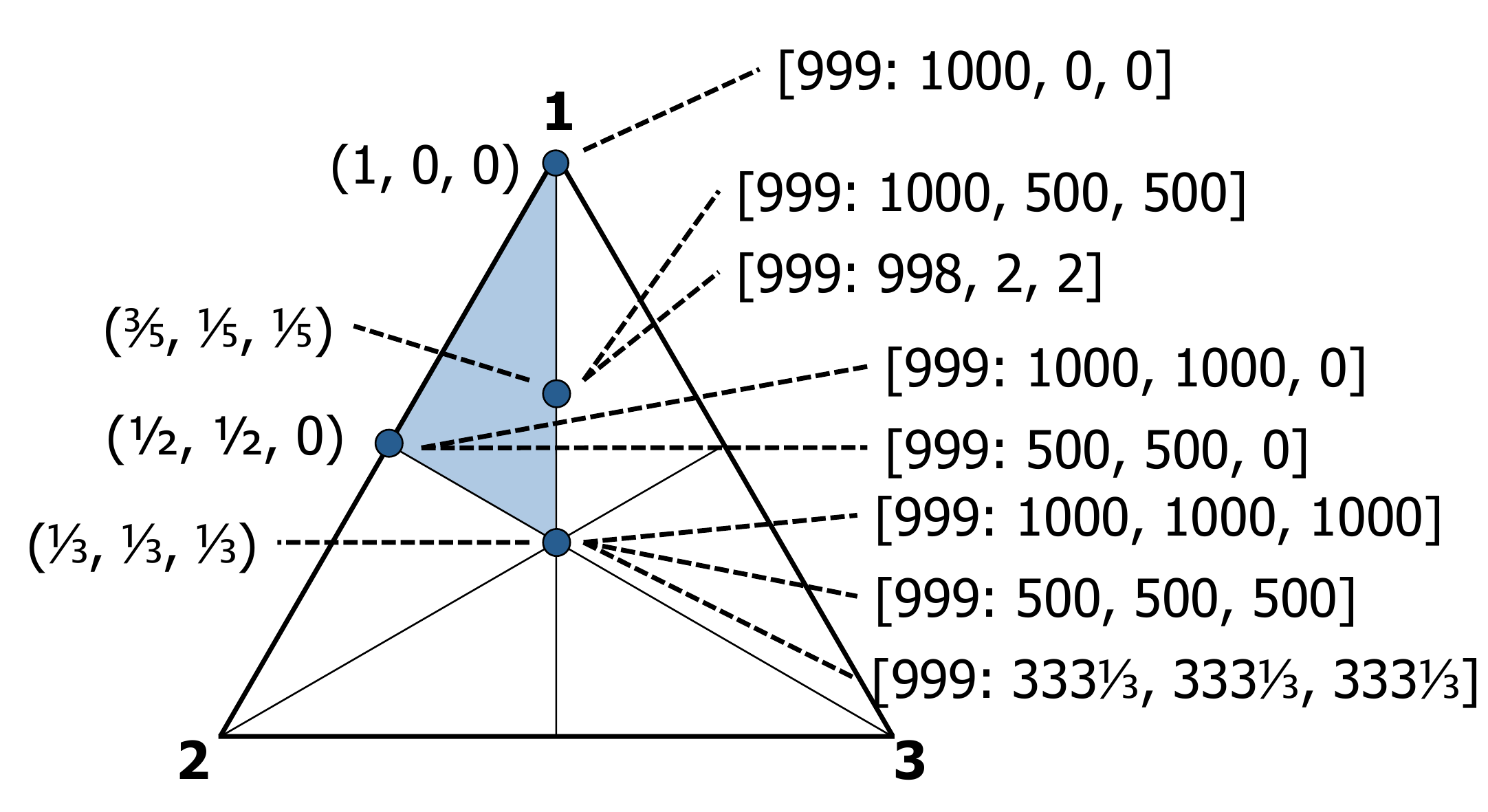}
\caption[short caption]{The games on three players, and their power indices.}
\label{fig:simplex}
\end{figure}
Because we focus on canonical weighted voting games, only the shaded part of the simplex contains games. These games are represented as dark dots. There are four dots, but, as we shall see, there are ten 3-player games. Two of these are degenerate, namely the game with no winning coalitions, and the game in which the empty set is the minimal winning coalition, so all coalitions are winning. Weighted representations for the other eight games are given on the right in the figure. There are only four distinct power indices corresponding to these games, they are indicated on the left.

In the remainder of this section, we will discuss the results obtained from some experiments that we have performed by implementing Algorithm \ref{alg:wvgoutput}, and the algorithm for $(\beta, \mathcal{G}_{\mathsf{wvg}}, \mathcal{L}_{\mathsf{weights}})$-PVGD that directly follows from it (where $\beta$ denotes the normalized Banzhaf index). There are various reasons for performing these experiments: First of all, we are interested in running our algorithm for some small choices of $n$ to see at what point our algorithm becomes intractable. A second goal of these experiments is to obtain some interesting statistics about the class of canonical weighted voting games (e.g., the number of weighted voting games on $n$ players). Thirdly, it we are interested in obtaining some statistics on the average optimal attainable error on a random instance, when we let the algorithm run to completion for small $n$. Lastly, we want to know about the error convergence rate of the algorithm for larger values of $n$, when
solving the problem to optimality is intractable.
More precisely, we want to gain insight in the following:

\begin{itemize}
\item the practical time-performance of the algorithm (for small $n$);
\item the average optimal attainable error on random instances (for small $n$);
\item the error-convergence behaviour of the algorithm (for larger $n$, when it becomes intractable to run the algorithm to completion);
\item obtaining the exact number of weighted voting games of $n$ players, in order to compare this to the theoretical bounds;
\item obtaining the number of weighted voting games for fixed numbers of players, as a function of the number of minimal winning coalitions.
\end{itemize}

In Section \ref{ch6sec1} we give some important information about the implementation of our algorithm. Section \ref{ch6sec2} describes our experiments. Lastly, in Section \ref{ch6sec3} we present the results of the experiments.

\subsection{Implementation details}\label{ch6sec1}
We have implemented Algorithm \ref{alg:wvgoutput} together with all of the optimization tricks described in Section \ref{improvingthealgorithm}. The programming language that we used is \emph{C}.

Execution of the algorithm encompasses solving a large number of linear programs. For doing this, we make use of the \emph{GNU Linear Programming Toolkit} \cite{glpk}. This is an open-source C library.

As said in the introduction of this section, our implementation solves the $(\beta, \mathcal{G}_{\mathsf{wvg}}, \mathcal{L}_{\mathsf{weights}})$-PVGD problem, where $\beta$ is the \emph{normalized} Banzhaf index. This means that for each weighted voting game that is output by our enumeration algorithm, we must invoke a procedure for computing the normalized Banzhaf index. The algorithm we use for this is simply the naive brute-force approach.

Two variants of the enumeration algorithm have been implemented: The first one uses the standard breadth-first approach, that sequentially generates all weighted voting games of $i$ minimal winning coalitions, for increasing $i$. The second one uses the depth-first method mentioned in Remark \ref{depthfirst} (in Section \ref{sec:thealgorithm}).

\subsection{Experiments}\label{ch6sec2}
We perform our experiments on a computer with an Intel Core2 Quad Q9300 2.50GHz CPU with 2GB SDRAM Memory. The operating system is Windows Vista. We compiled our source code using gcc 3.4.4, included in the DJGPP C/C++ Development System. We compiled our code with the \emph{--O3} compiler flag.

For doing the experiments, we need input data: instances that we use as input for the algorithm. An instance is a target banzhaf index for a canonical weighted voting game, i.e., a point $p$ in the unit simplex such that $p_i \geq p_j$ if $i < j$, for all $i,j$ between $1$ and $n$. Our instances therefore consist of samples of such vectors that were taken uniformly at random. These samples are generated according to the procedure described in \cite{uniformsimplex}.

The experiments are as follows:
\begin{description}
\item[Experiment 1:] For up to 8 players, we measured the CPU time it takes for the enumeration algorithm to output all games, for both the breadth-first and the depth-first method. From these experiments we obtain the exact number of canonical weighted voting games of $n$ players for all $n$ between 1 and 8. We also measure the additional runtime that is necessary when we include the computation of the Banzhaf index in the algorithm.
\item[Experiment 2:] We use the enumeration algorithm to compute for all $n$ with $1 \leq n \leq 8$ and all $m$ with $0 \leq m \leq \binom{n}{\lfloor n/2 \rfloor}$, the exact number of canonical weighted voting games on $n$ players with $m$ minimal winning coalitions.
\item[Experiment 3:] For $n$ between 1 and 7, we compute for 1000 random instances the \emph{average optimal error}. That is, the average error that is attained out of 1000 random instances (i.e., uniform random vectors in the $(n-1)$-dimensional unit-simplex), when the algorithm is allowed to run to completion on these instances. We also report the worst error that is attained among these 1000 instances. The error function we use is the square root of the sum of squared errors, as stated in Definition \ref{pvgdproblemdef}. The reason for using this specific error measure is because it has a nice geometric interpretation: it is the Euclidean distance between the target (input) vector and the closest point in the unit simplex that is a normalized Banzhaf index of a weighted voting game.
\item[Experiment 4:] For $n \in \{10, 15, 20\}$, we measure the error-convergence behaviour of the algorithm: the Euclidean error as a function of the amount of time that the algorithm runs. We again do this experiment for both the breadth-first and the depth-first version of the algorithm. For each of these three choices of $n$, we perform this experiment for $10$ random instances, and for each instance we allow the algorithm to run for one minute.
\end{description}

\subsection{Results}\label{ch6sec3}
For Experiment 1, the runtimes are given in Figure \ref{graph1}. From the graph we see that for all four versions of the algorithm, there is relatively not much difference in the runtimes. This means that the inclusion of the Banzhaf index computation procedure does not add a significant amount of additional runtime. Nonetheless, one should not forget that these results are displayed on a logarithmic scale. When we compare the runtimes for 8 players with each other for example, we see that the runtime of the depth-first search version without Banzhaf index computation is 21 minutes, while it is 26 minutes when we include the computation of the Banzhaf index into the algorithm. When we use the breadth-first search approach instead, the runtime is only 16 minutes. In general, the breadth-first search method is a lot faster than the depth-first search method.

\begin{figure}
\center
\includegraphics[scale=1.0,angle=0]{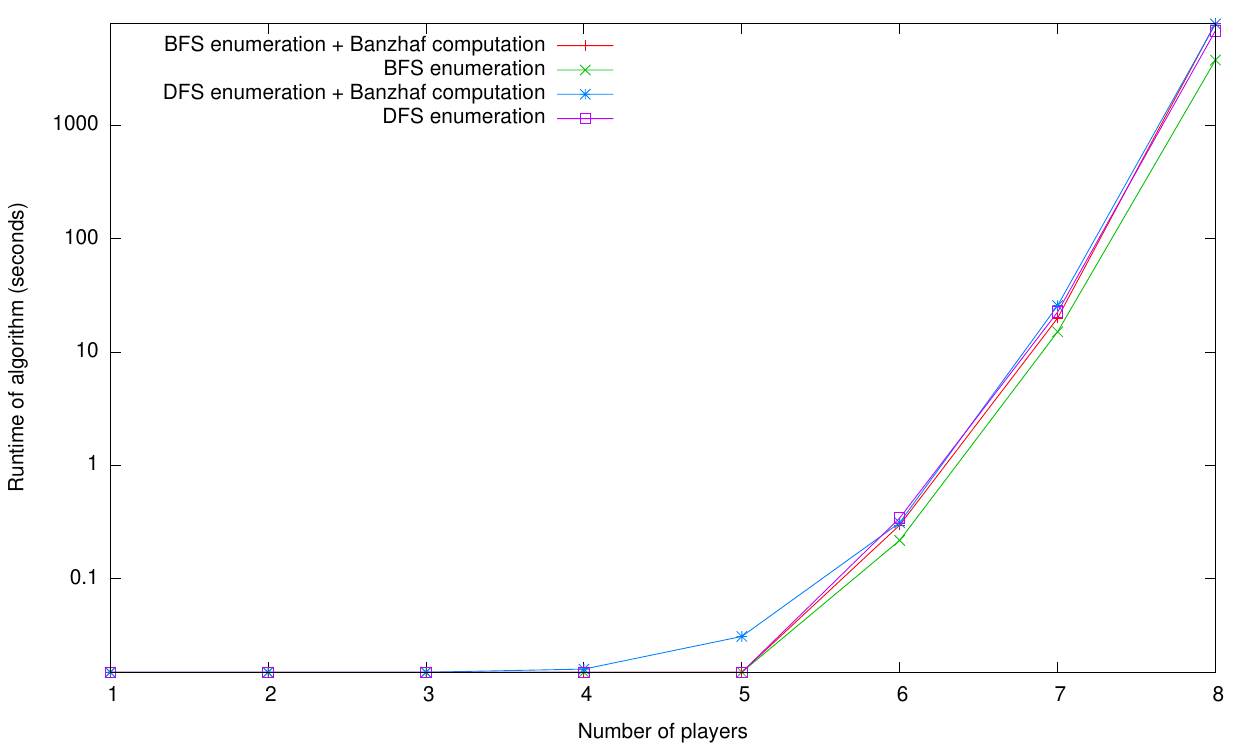}
\caption[Runtimes of various versions of Algorithm \ref{alg:wvgoutput}, for $1 \leq n \leq 8$.]{Runtimes of Algorithm \ref{alg:wvgoutput} for 1 to 8 players, for both the breadth-first search and the depth-first search variant of the algorithm, both with and without the Banzhaf index computation procedure included.}
\label{graph1}
\end{figure}

The number of canonical weighted voting games on $n$ players, for $1 \leq n \leq 8$, is displayed in Figure \ref{graph2}. Even for these small values of $n$, we can already clearly see the quadratic curve of the graph on this log-scale, just as the theoretical bounds from Section \ref{cardinalities} predict. In Table \ref{numberofgames}, we state the exact numbers of canonical weighted voting games on $n$ players as numbers, for $1 \leq n \leq 8$.

\begin{figure}
\center
\includegraphics[scale=1.0,angle=0]{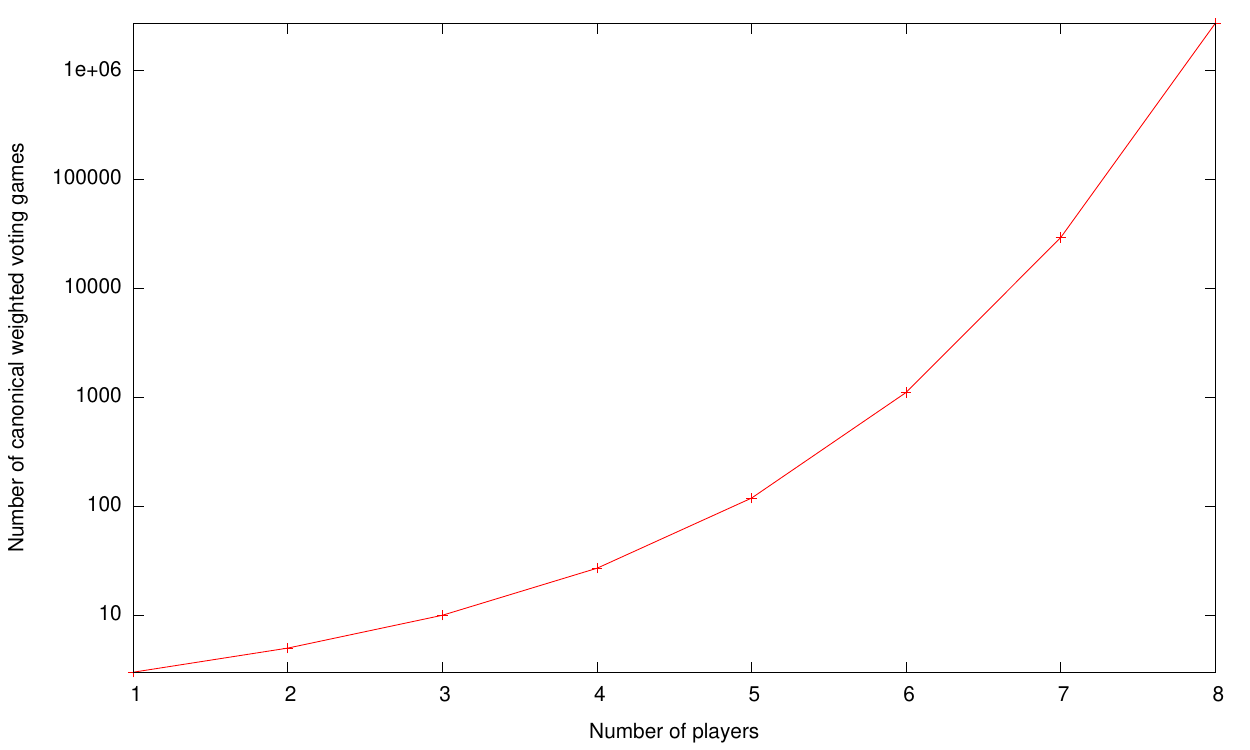}
\caption{The number of canonical weighted voting games on $n$ players, for $1 \leq n \leq 8$.}
\label{graph2}
\end{figure}

\begin{table}
\caption{Exact values for the number of weighted voting games on $n$ players, for $1 \leq n \leq 8$.}
\label{numberofgames}
\begin{center}
\begin{tabular}{|c|c|}
$n$ & $|\mathcal{G}_{\mathsf{cwvg}}(n)|$ \\
\hline \hline
1 & 3\\
\hline
2 & 5\\
\hline
3 & 10\\
\hline
4 & 27\\
\hline
5 & 119\\
\hline
6 & 1113\\
\hline
7 & 29375\\
\hline
8 & 2730166\\
\hline
\end{tabular}
\end{center}
\end{table}

For Experiment 2, the results are displayed in Figure \ref{graph3}. Note that on the vertical axis we have again a log-scale.
We see that for each of these choices of $n$, most of the canonical weighted voting games have a relatively low number of minimal winning coaltions relative to the maximum number of winning coalitions $\binom{n}{\lfloor n/2 \rfloor}$.

\begin{figure}
\center
\includegraphics[scale=1.0,angle=0]{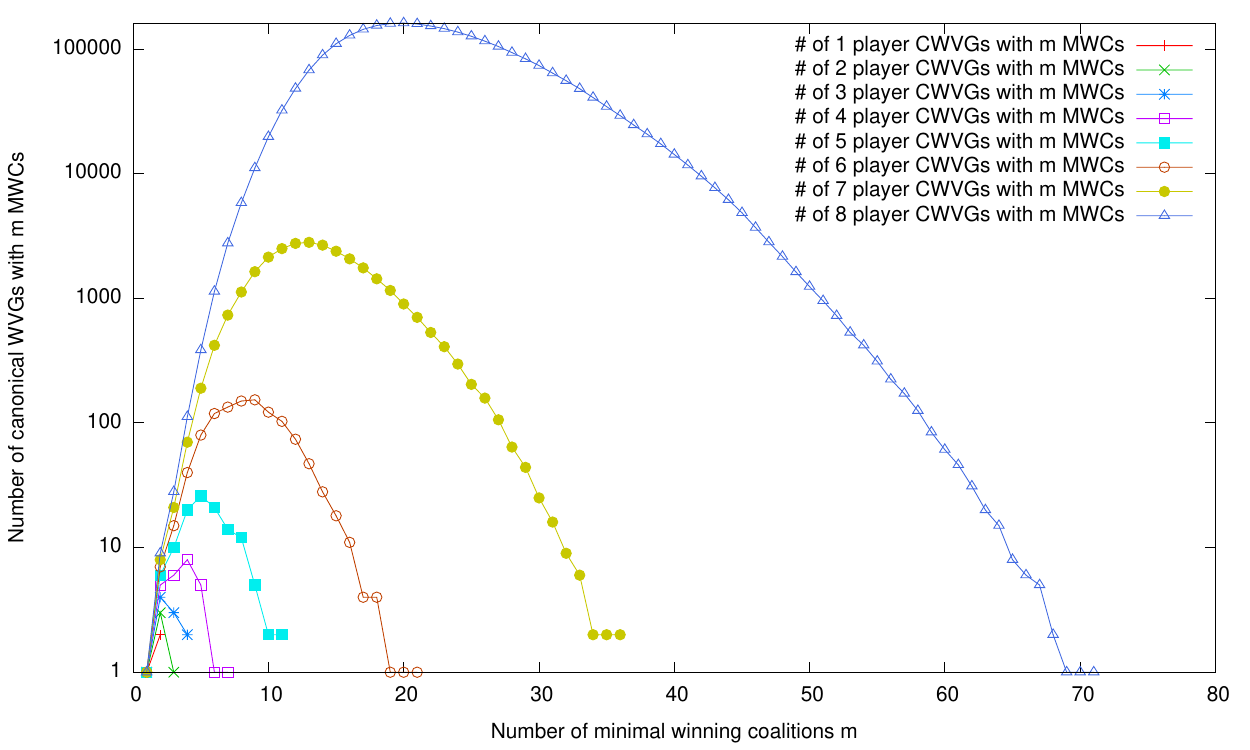}
\caption[The number of canonical weighted voting games of $n$ players with $m$ MWCs.]{The number of canonical weighted voting games (y-axis) on $n$ players, for $1 \leq n \leq 8$, with $m$ minimal winning coalitions (x-axis).}
\label{graph3}
\end{figure}

The Euclidean errors computed in Experiment 3 are displayed in Figure \ref{graph4}.
We see that the errors decrease as $n$ gets larger. We also see that the worst case optimal error can be much worse than the average case. We want to emphasize that these are results computed over only 1000 random instances. Therefore, these worst case optimal errors serve only as a lower bound for the worst case optimal error over all possible instances.
\begin{figure}
\center
\includegraphics[scale=1.0,angle=0]{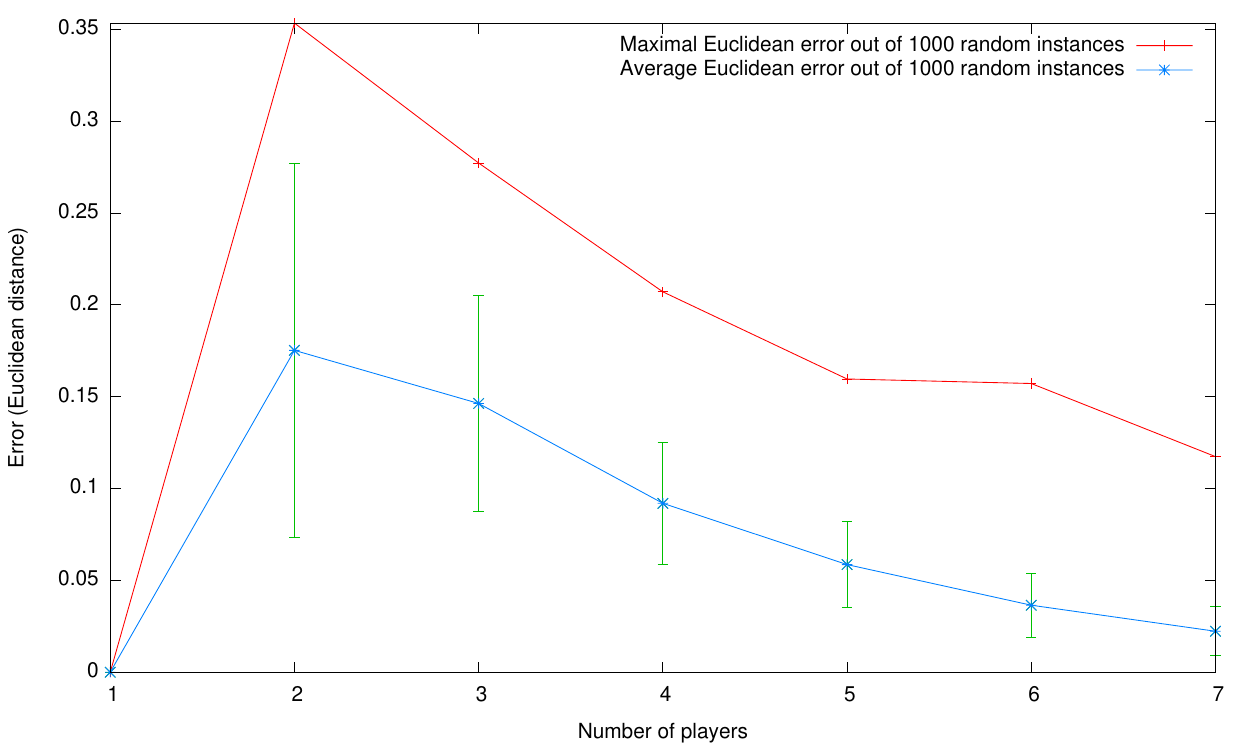}
\caption[Optimal Euclidean error of 1000 random $n$ player instances, for $1 \leq n \leq 7$.]{Optimal Euclidean error of 1000 random $n$ player instances, for $1 \leq n \leq 7$. The error bars indicate one standard deviation.}
\label{graph4}
\end{figure}

For Experiment 4, we see no possibility for a meaningful or interesting visualisation of its results.
Experiment 4 confirms to us that this enumeration-approach of solving PVGD problems quickly becomes impractical as $n$ gets larger.
Our hopes were that the anytime-property of the algorithm would account for a quick convergence to a low (but not necessarily optimal) error; even for large values of $n$. It turns out that this is not the case. In all cases (i.e., for $n = 10$, $n = 15$ and $n = 20$, for all of the $10$ random instances), the error-convergence is high during approximately the first second that the algorithm runs. After that, the frequency by which improvements in the error occur, seems to decrease exponentially. Moreover, it holds without exception that after the first second, the improvements are only tiny.
The average euclidean errors obtained after letting the algorithm run for one minute are as follows:
\begin{itemize}
\item For $n = 10$, after one minute, the average euclidean error over the 10 instances was 0.055234 for the breadth-first variant, and 0.1705204 for the depth-first variant.
\item For $n = 15$, after one minute, the average euclidean error over the 10 instances was 0.0983193 for the breadth-first variant, and 0.2018266 for the depth-first variant.
\item For $n = 20$, after one minute, the average euclidean error over the 10 instances was 0.1475115 for the breadth-first variant, and 0.2399217 for the depth-first variant.
\end{itemize}

From this, we see that for $n = 10$, the breadth-first search method still gives us reasonably nice results within a minute, but when we increase the number of players to $15$ and $20$, we see that the results quickly get worse. Especially when we compare the results to the expected average optimal error (that we obtain by extrapolation of the results of Experiment 3).

Another interesting observation is that these errors for the depth-first variant are much worse than the errors for the breadth-first variant. An explanation for this is that the Banzhaf indices of the generated games are scattered more evenly across the unit simplex in the case of the breadth-first variant: We expect the depth-first variant to enumerate a lot of games for which the Banzhaf indices are close to each other, due to the cover relation of $(\mathcal{G}_{\mathsf{cwvg}}(n), \preceq_{\mathsf{MWC}})$.

A final comment we would like to make is that when $n$ gets larger, the output rate of the enumeration algorithm goes down. Of course, this is explained by the fact that many of the operations in the algorithm must now be performed on games with more players. Especially this slowdown is caused by the computation of the Banzhaf index that is done for every game. In our current implementation, computing the Banzhaf index takes time exponential in $n$.

In general, our current implementation is crude: many procedures in this implementation are still far from optimal. We expect that it is possible to attain a significant improvement in the performance of this algorithm by optimizing the code.

\section{Conclusions \& future work}\label{sec:conclusions}
In this paper, we have derived an \emph{exact} algorithm for solving power index weighted voting game design problems.
We have shown that such a problem is always solvable for any class of games, but the guarantee on the worst-case runtime that we can give is unfortunately only doubly exponential. For the important case of weighted voting games, we have derived an anytime method that runs in exponential time, and we have developed various additional techniques that we can use to speed this algorithm up.

This algorithm is based on an enumeration procedure for the class of weighted voting games: it works by simply enumerating every game, and verifying for each game whether it lies closer to the target power index than the games that we encountered up until that point. For this reason, the algorithm has the anytime-property: as we run this algorithm for a longer period of time, the algorithm enumerates more games, and the quality of the solution will improve.

Also, due to the genericity of enumeration, we can use our algorithm not only to solve power index voting game design problems: we can use it to solve any other voting game design problem as well. The only thing we have to adapt is the error-function of the algorithm (i.e., the part of the algorithm that checks the property in question for each of the games that the enumeration procedure outputs); the enumeration procedure does not need to be changed.

Finally, we implemented a simple, non-optimized version of the algorithm in order to do some experiments and obtain some statistical information about the class of weighted voting games. We have computed some exact values for the number of canonical weighted voting games on $n$ players with $m$ minimal winning coalitions, for small choices of $n$, and every $m$. We have seen that even for small $n$, it is already obvious from the experimental results that the number of weighted voting games grows quadratically on an exponential scale, precisely according to the known asymptotic bounds.

We measured the runtime of the algorithm, and observed that running the algorithm to completion becomes intractable at approximately $n = 10$ (on the computer that we performed the experiments with, we estimate that it takes a month to run the algorithm to completion for $n = 9$). Lastly, for larger values of $n$, our algorithm (or at least our current implementation) is of little use for practical purposes because the error does not converge as quickly as we would want to. We think that we can attain a significant speedup by optimizing the code, and by using better linear programming software.

Note that in most real-life examples, the number of players in a weighted voting game is rather small: usually 10 to 50 players are involved. For future work, the goal is to get this algorithm to yield good results within a reasonable amount of time when the number of players is somewhere in this range. It would already be interesting to be able to solve the problem for ten players, as we are not aware of any enumerations of ten player canonical weighted voting games. However, we concluded that the current implementation of the algorithm is not yet fast enough to be able to handle ten players. Optimistic extrapolation tells us that this would take tens of years; pessimistic extrapolation gives us thousands of years. However, the current implementation is not really efficient, and we have hope that with some future insights, together with careful computer programming, enumerating weighted voting games for ten players or more will be within scope.

We think that it will be interesting to study in more depth the partial order we introduced in this paper, both from a from a computational perspective and from a purely mathematical perspective.
One possible prospect is the following. With regard to weighted voting game design problems, we suspect that it is possible to prune a lot of ``areas'' in this partial order: Careful analysis of the partial order and its properties might lead to results that allow us to construct an enumeration algorithm that \emph{a priori} discards certain (hopefully large) subsets of weighted voting games.

We are moreover interested to see how an algorithm performs that searches through the partial order in a greedy manner, or what will happen if we use some other (possibly heuristic) more intelligent methods to search through the partial order. We wonder if it is possible to use such a search method while still having an optimality guarantee or approximation guarantee on the quality of the solution. Lastly, we can also consider the ideas presented here as a postprocessing step to existing algorithms. In other words, it might be a good idea to first run the algorithm of \cite{shaheeninverse} or \cite{azizinverse} in order to obtain a good initial game. Subsequently, we can try to search through the ``neighborhood'' of the game to find improvements, according to the partial order introduced in this paper, .

Lastly, some related questions for which it would be interesting to obtain an answer are about the computational complexity of the power index voting game design problem, and also about the polynomial-time-approximability of the problem. It is quite straightforward to see that the decision version of this problem is in most cases in $\mathsf{NP}^{\mathsf{\#P}}$ (and therefore in $\mathsf{PSPACE}$), as one could nondeterministically guess a weight vector, and subsequently use a $\mathsf{\#P}$-oracle to obtain the particular power index of interest.\footnote{All power indices that have been proposed and that we have encountered are known to be in $\mathsf{\#P}$} On the other hand, at the moment we do not have any ideas on how to prove hardness for this problem for any complexity class whatsoever. It seems a challenge to come up with a polynomial-time reduction from any known computational problem that is hard for any nontrivial complexity class. Also, on questions related to approximability of PVGD problems we 
currently do not have an answer.

\paragraph{Acknowledgements}
We thank Fabi\'{a}n Riquelme and Andreas Polym\'{e}ris for pointing out a problem in a preliminary version of this paper (See Section \ref{sec:outputpolynomial}).

\bibliography{vgdjournal}

\end{document}